\documentclass[11pt, a4paper]{article}
\usepackage[pagebackref=false,citecolor=blue,colorlinks=false,linkcolor=blue,bookmarks=false]{hyperref}
\usepackage[margin=1in]{geometry}
\usepackage[utf8]{inputenc}
\usepackage{amsmath}
\usepackage{amssymb}
\usepackage{amsthm}
\usepackage{fancybox}
\usepackage[english]{babel}
\usepackage{xcolor}
\usepackage{algorithm,algorithmicx}
\usepackage[noend]{algpseudocode}

\usepackage{tikz,subfigure}
\usepackage[active]{srcltx}

\newtheorem{thm}{Theorem}[section] 
\newtheorem{theorem}[thm]{Theorem}
\newtheorem{fact}[thm]{Fact}
\newtheorem{remark}[thm]{Remark}
\newtheorem{definition}[thm]{Definition}

\newtheorem{lemma}[thm]{Lemma}

\newtheorem{cor}[thm]{Corollary}

\numberwithin{equation}{section}

\newcommand{\RRR}{\mathbb{R}}
\newcommand{\CCC}{\mathbb{C}}

\newcommand{\calL}{\mathcal{L}}
\newcommand{\LP}{\widetilde{L}}
\newcommand{\AP}{\widetilde{A}}
\newcommand{\DP}{\widetilde{D}}

\newcommand{\xp}{\widetilde{x}}
\newcommand{\bp}{\widetilde{B}}
\newcommand{\fp}{\widetilde{F}}
\newcommand{\pb}{p}

\newcommand{\Otil}{\widetilde{O}}

\newcommand{\rot}{\intercal}

\newcommand{\nL}{\mathcal{L}}

\newcommand{\vol}{\mathsf{Vol}}

\newcommand{\hmod}{\hspace{-5pt}\mod}
\newcommand{\eps}{\varepsilon}

\newcommand{\maxlin}{\textsf{MAX-2-LIN($k$)}}
\newcommand{\weit}{b}
\newcommand{\ists}{\mathcal{I}}
\newcommand{\jsts}{\mathcal{J}}
\newcommand{\cst}{c}
\newcommand{\asn}{\phi}
\newcommand{\una}{\bot}
\newcommand{\pnt}{p}
\newcommand{\kpa}{\nu}
\newcommand{\lev}{\ell}

\newcommand{\IH}{\widehat{\ists}}
\newcommand{\GH}{\widehat{G}}

\newcommand{\calG}{\mathcal{U}}

\newcommand{\calU}{\mathcal{U}}

\def\ceil#1{\left\lceil #1 \right\rceil}

\def\prob#1#2{\mathbb{P}_{#1}\left[ #2 \right]}

\def\expec#1#2{{\mathbb{E}}_{#1}\left[ #2 \right]}

\newcommand{\E}{\mbox{{\bf E}}}

\let\conj\overline
\def\norm#1{\left\| #1 \right\|}
\def\trace#1{\mathrm{tr} \left(#1 \right)}

\newcommand{\lambdamin}{\lambda_{\mathrm{min}}}
\newcommand{\lambdamax}{\lambda_{\mathrm{max}}}

\def\defeq{\stackrel{\mathrm{def}}{=}}
\def\setof#1{\left\{#1  \right\}}
\newcommand{\kh}[1]{\left(#1\right)}
\def\sizeof#1{\left|#1  \right|}

\def\union{\cup}

\makeatletter
\tikzset{
  @arc through/.style 2 args={
    to path={
      \pgfextra
        \pgfextract@process\pgf@tostart{\tikz@scan@one@point\pgfutil@firstofone(\tikztostart)\relax}%
        \pgfextract@process\pgf@tothrough{\tikz@scan@one@point\pgfutil@firstofone#1}%
        \pgfextract@process\pgf@totarget{\tikz@scan@one@point\pgfutil@firstofone(\tikztotarget)\relax}%
        \pgfextract@process\pgf@topointMidA{\pgfpointlineattime{.5}{\pgf@tostart}{\pgf@tothrough}}%
        \pgfextract@process\pgf@topointMidB{\pgfpointlineattime{.5}{\pgf@totarget}{\pgf@tothrough}}%
        \pgfextract@process\pgf@tocenter{%
          \pgfpointintersectionoflines{\pgf@topointMidA}
            {\pgfmathrotatepointaround{\pgf@tothrough}{\pgf@topointMidA}{90}}
            {\pgf@topointMidB}{\pgfmathrotatepointaround{\pgf@tothrough}{\pgf@topointMidB}{90}}}%
        \pgfcoordinate{arc through center}{\pgf@tocenter}%
        \pgfpointdiff{\pgf@tocenter}{\pgf@tostart}%
        \pgfmathveclen@{\pgfmath@tonumber\pgf@x}{\pgfmath@tonumber\pgf@y}%
        \edef\pgf@toradius{\pgfmathresult pt}
        \pgfmathanglebetweenpoints{\pgf@tocenter}{\pgf@tostart}%
        \let\pgf@tostartangle\pgfmathresult
        \pgfmathanglebetweenpoints{\pgf@tocenter}{\pgf@totarget}%
        \let\pgf@toendangle\pgfmathresult
        \ifdim\pgf@tostartangle pt>\pgf@toendangle pt\relax
          \pgfmathsetmacro\pgf@tostartangle{\pgf@tostartangle-360}%
        \fi
        #2%
          \pgfmathsetmacro\pgf@toendangle{\pgf@toendangle-360}%
        \fi
      \endpgfextra
      arc [radius=+\pgf@toradius, start angle=\pgf@tostartangle, end angle=\pgf@toendangle] \tikztonodes
    }},
  arc through ccw/.style={@arc through={#1}{\iffalse}},
  arc through cw/.style={@arc through={#1}{\iftrue}},
}
\makeatother

\newenvironment{fminipage}%
  {\begin{Sbox}\begin{minipage}}%
  {\end{minipage}\end{Sbox}\fbox{\TheSbox}}

\title{Hermitian Laplacians and a Cheeger inequality \\
 for the Max-2-Lin problem\label{sec:preliminaries}}
\author{Huan Li\footnote{School of Computer Science, Fudan University, China} \and He Sun\footnote{School of Informatics, the University of Edinburgh, UK} \and Luca Zanetti\footnote{Department of Computer Science and Technology, the University of Cambridge, UK}}
\date{}

\begin{document}

\maketitle

\begin{abstract}

We study spectral approaches for the \maxlin\ problem, in which we are given 
a system of $m$ linear equations of the form 
$
x_i - x_j \equiv \cst_{ij}\hmod k$,  
and required    to find an assignment to the $n$ variables $\{x_i\}$ that maximises the total number of satisfied equations. 

We consider Hermitian Laplacians related to this problem, and 
prove a Cheeger inequality  that relates the smallest eigenvalue of a Hermitian Laplacian to the maximum number of satisfied equations of a \maxlin\ instance $\ists$.
We develop  an  $\widetilde{O}(kn^2)$ time\footnote{The notation $\widetilde{O}(\cdot)$ suppresses poly-logarithmic factors in $n$, $m$, and $k$.} algorithm that, for any $(1-\eps)$-satisfiable instance, produces an assignment satisfying a
 $\left(1 - O(k)\sqrt{\eps}\right)$-fraction of equations.  We also present a subquadratic-time algorithm that, when the graph associated with $\ists$ is an expander, produces an assignment satisfying a $\left(1- O(k^2)\eps \right)$-fraction of the equations. Our Cheeger inequality and first algorithm can be seen as generalisations of the Cheeger inequality and algorithm for \textsf{MAX-CUT} developed by Trevisan.


\end{abstract}

\thispagestyle{empty}

\setcounter{page}{0}
\newpage

\section{Introduction}

In the \maxlin\ problem, we are given a system of $m$ linear equations of the form 
\begin{equation}\label{eq:general}
u_i - v_i \equiv \cst_i\hmod k
\end{equation}
where $u_i, v_i \in \setof{x_1,\ldots,x_n}$
and each equation has weight $\weit_i$.
The objective is to find an assignment to the variables $x_i$ that maximises the total weight of satisfied equations.  As an important case of  Unique Games~\cite{feigelovasz92,Kho02a}, the \maxlin\ problem has been extensively studied in theoretical computer science.
 This problem is known 
to be \textsf{NP}-hard to  approximate within a ratio of $11/12+\delta$ for any constant $\delta> 0$~\cite{feige04,Hastad01}, and it is conjectured to be hard to distinguish between \maxlin\ instances for which a $(1-\eps)$-fraction of equations can be satisfied versus  
  instances for which only an $\eps$-fraction can be satisfied~\cite{KhotKMO07}. 
  On the algorithmic side, there has been a number of  LP and SDP-based algorithms proposed for the \maxlin\ problem~(e.g., \cite{Kho02a,trevisan05,charikar06,gupta06}), and  
  the   case of $k=2$, which corresponds to the  classical \textsf{MAX-CUT} problem for undirected graphs~\cite{GW95, Karp72},
  has been widely studied over the past fifty years.

  


In this paper we investigate  efficient spectral  algorithms  for \maxlin. For any \maxlin\ instance $\ists$ with $n$ variables, we express $\ists$ by a Hermitian Laplacian matrix $L_\ists\in\CCC^{n\times n}$, and analyse the spectral properties of $L_{\ists}$. In comparison to the well-known Laplacian matrix for undirected graphs~\cite{chung1}, complex-valued entries in $L_{\ists}$ are able to express directed edges in the graph associated with $\ists$, and at the same time ensure that  all the eigenvalues of $L_{\ists}$ are real-valued. We demonstrate the power of our Hermitian Laplacian matrices by relating the maximum number of satisfied equations of $\ists$ to the spectral properties of $L_{\ists}$. In particular, we develop a Cheeger inequality that relates partial assignments of $\ists$ to $\lambda_1(L_{\ists})$, the smallest eigenvalue of $L_{\ists}$. 
Based on a recursive application of the algorithm behind our Cheeger inequality, as well as a spectral sparsification procedure for  \maxlin\ instances,  we present an approximation algorithm for \maxlin\ that runs in $\widetilde{O}(k\cdot n^2)$  time. Our algorithm is easy to implement, and is 
  significantly faster than most SDP-based algorithms for this problem in the literature, while achieving similar guarantees for constant values of $k$.  
The formal statement of our result is as follows:

\begin{theorem}
\label{thm:main1}
There is an $\widetilde{O}(k\cdot n^2)$-time algorithm such that, for any given \maxlin\ instance $\ists$ with optimum  $1 - \eps$, the algorithm returns 
    an assignment $\asn$ satisfying at least 
a $(1 - O(k) \sqrt{\eps})$-fraction of the equations\footnote{
An instance $\ists$ has optimum $1-\eps$, if the maximum fraction of the total weights of satisfied equations is $1-\eps$.}.
\end{theorem} 

Our result can be viewed as a generalisation of  the \textsf{MAX-CUT} algorithm by Trevisan~\cite{Trevisan09}, who derived a Cheeger inequality that relates the value of the maximum cut to the smallest eigenvalue of an undirected graph's adjacency matrix. The proof of Trevisan's Cheeger inequality, however, is based on constructing sweep sets in $\RRR$, while in our setting  constructing
 sweep sets in $\CCC$ is needed, as the underlying graph defined by $L_{\ists}$ is directed and eigenvectors of $L_{\ists}$ are in $\CCC^n$. 
 The other difference between our result and the one in \cite{Trevisan09} is that
 the goal of the \textsf{MAX-CUT} problem is to find a \emph{bipartition} of the vertex set, while for the \maxlin\ problem we need to
 use an eigenvector to find $k$ vertex-disjoint subsets, which corresponds to subsets of variables assigned to the same value. 

 Our approach also shares some similarities with the one by Goemans and Williamson~\cite{GW04}, who presented a $0.793733$-approximation algorithm for \textsf{MAX-2-LIN(3)} based on Complex Semidefinite Programming. The objective function of their SDP relaxation is, in fact, exactly the quadratic form of our Hermitian Laplacian matrix $L_{\ists}$, although this matrix was not explicitly defined in their paper. In addition, their rounding scheme divides the complex unit ball into $k$ regions according to the angle with a random vector, which is  part of our rounding scheme as well. Therefore, if  one views Trevisan's work~\cite{Trevisan09}  as a spectral analogue to the celebrated SDP-based algorithm for \textsf{MAX-CUT} by Goemans and Williamson~\cite{GW95}, our result can be seen as a spectral analogue to the  Goemans and Williamson's algorithm for \maxlin. 
 
 We further prove that, when the undirected graph associated with a \maxlin\ instance is an expander,
the approximation ratio from Theorem~\ref{thm:main1}
can be improved. Our result is formally stated as follows:

\begin{theorem}\label{thm:res2}
    Let $\ists$ be an instance of \maxlin\ on a $d$-regular graph with $n$ vertices\
     and suppose its optimum is $1 - \eps$.
     There is an $\Otil\left(nd + \frac{n^{1.5}}{k\sqrt{\eps}}\right)$-time algorithm that
    returns an
    assignment $\asn: V\to [k]$
    satisfying at least a
    \begin{align}\label{eq:thm2eq}
        1 - O(k^2)\cdot \frac{\eps}{\lambda_2^3(\calL_{\calG})}
    \end{align}
    fraction of equations in $\ists$,
    where $\lambda_2(\calL_{\calG})$
    is the second smallest eigenvalue of the normalised Laplacian matrix of the underlying
    undirected graph $\calG$.
    \end{theorem}

Our technique is similar to the one by Kolla~\cite{kolla11},
which was used to show that solving the \maxlin\ problem on expander graphs is easier.
In \cite{kolla11}, a \maxlin\ instance is represented by the label-extended graph, and the algorithm  is
based on an exhaustive search
in a subspace spanned by eigenvectors associated with eigenvalues close to $0$.
When the underlying graph of the \maxlin\ instance has good expansion,
this subspace is of dimension $k$.
Therefore, the exhaustive search runs in
time $O\left(2^k + \mathrm{poly}(n\cdot k)\right)$, which is polynomial-time when $k = O(\log n)$.
Comparing with the work in \cite{kolla11}, we show that, 
 when the underlying graph has good expansion, the eigenvector associated with the smallest eigenvalue $\lambda_1(\calL_\ists)$ of the Hermitian Laplacians suffices to give a good approximation.  
We notice that Arora et al.~\cite{stoc/AroraKKSTV08} already showed that, for expander graphs, it is possible to 
satisfy a $1-O(\eps \log(1/\eps))$ fraction of equations in polynomial time without any dependency on $k$. Their algorithm is based on
an SDP relaxation.

%
%


\paragraph{Other related work.}
There are many research results for the \maxlin\ problem~(e.g., \cite{Kho02a,trevisan05,charikar06,gupta06}), and  we briefly discuss the ones most closely related to our work. 
For the \maxlin\ problem and Unique Games, spectral techniques are usually employed to analyse the Laplacian matrix of the so-called Label-Extended graphs. Apart from the above-mentioned result~\cite{kolla11}, Arora, Barak and Steurer~\cite{abs}
 obtained an  $\mathrm{exp}\left((kn)^{O(\eps)}\right)\mathrm{poly}(n)$-time algorithm for Unique Games, whose algorithm makes use of Label-Extended graphs as well.
 We also  notice that  the adjacency matrix corresponding to our Hermitian Laplacian was  considered by Singer~\cite{singer11} in relation to an angular synchronisation problem. The connection between the eigenvectors of such matrix and the \maxlin\ problem was also mentioned, but without offering formal approximation guarantees.

\section{Hermitian Matrices for \maxlin \label{sec:preliminiares}}

We can write an instance of \maxlin\ by $\ists = (G, k)$, where $G = (V,E,\weit,\cst)$ denotes
a directed graph with an edge weight function $\weit: E\to \mathbb{R}^+$
and an edge color function $\cst: E\to [k]$, where $[k]\defeq \setof{0,1,\ldots,k-1}$.
More precisely, every equation $u_i - v_i \equiv \cst_i \hmod k$ with weight $\weit_i$
corresponds to a directed edge $(u_i, v_i)$ with weight $\weit(u_i,v_i) = \weit_{u_iv_i} = \weit_i$ and color $\cst(u_i,v_i) = \cst_{u_iv_i} = \cst_i$.
In the rest of this paper, we will assume that $G$ is weakly connected, and write $u\leadsto  v$ if there is a directed edge from $u$ to $v$.  The conjugate transpose of any vector $x\in\CCC^n$ is denoted by $x^*$.

We define the Hermitian adjacency matrix $A_\ists\in\CCC^{n\times n}$ for instance $\ists$ by 
\begin{align}
    (A_\ists)_{uv} \defeq
    \begin{cases}
        \weit_{uv}\omega_k^{\cst_{uv}} & u\leadsto v, \\
        \weit_{vu}\conj{\omega_k}\,^{\cst_{vu}} & v\leadsto u, \\
        0 & \mathrm{otherwise},
    \end{cases}
\end{align}
where $\omega_k = \exp\left(\frac{2\pi i}{k}\right)$ is the complex
$k$-th root of unity,
and $\conj{\omega_k} = \exp\left(-\frac{2\pi i}{k}\right)$ is its conjugate.
We define the degree-diagonal matrix $D_\ists$ by
$(D_\ists)_{uu} = d_u$ where $d_u$ is the weighted degree given by
\begin{align}
    d_u \defeq \sum_{u\leadsto v} \weit_{uv}
    + \sum_{v\leadsto u} \weit_{vu}.
\end{align}
The Hermitian Laplacian matrix is then defined by $L_{\ists} = D_{\ists} - A_{\ists}$,
and the corresponding normalised Laplacian matrix by
$\calL_{\ists} = D_{\ists}^{-1/2} L_{\ists} D_{\ists}^{-1/2} = I - D_{\ists}^{-1/2} A_{\ists} D_{\ists}^{-1/2}$. The eigenvalues of 
any matrix $A$ are expressed by $\lambda_1(A)\leq \ldots\leq \lambda_n(A)$.
The quadratic forms of $L_{\ists}$ can be related to the corresponding instance of  \maxlin\ by the following lemma.

\begin{lemma}\label{lem:qL}
    For any vector $x \in \mathbb{C}^n$, we have
    \begin{align}
        x^* L_\ists x = \sum_{u\leadsto v} \weit_{uv} \norm{x_u - \omega_k^{\cst_{uv}} x_v}^2
    \end{align}
    and
    \begin{align}
        x^* L_\ists x = 2\sum_{u\in V} d_u \norm{x_u}^2 - \sum_{u\leadsto v} \weit_{uv} \norm{x_u + \omega_k^{\cst_{uv}} x_v}^2.
    \end{align}
\end{lemma}
\begin{proof}
    For any vector $x \in \mathbb{C}^n$,
    we can write
    \begin{align}
        x^* A_\ists  x &=
        \sum_{u\leadsto v}
        \weit_{uv} \kh{ \conj{x_u} \omega_k^{\cst_{uv}} x_v
        + \conj{x_v}\,\conj{\omega_k}^{\cst_{uv}} x_u } \notag \\
        &= - \sum_{u\leadsto v}\weit_{uv}  \kh{ 
            \kh{ \conj{x_u} - \conj{x_v}\, \conj{\omega_k}\,^{\cst_{uv}}}
            \kh{ x_u - \omega_k^{\cst_{uv}} x_v }
    - \norm{x_u}^2 - \norm{x_v}^2 } \notag \\
    &= \sum_{u\in V} d_u \norm{x_u}^2-
    \sum_{u\leadsto v} \weit_{uv} \norm{x_u - \omega_k^{\cst_{uv}} x_v}^2.
    \end{align}
    We can also write
    \begin{align}
        x^* A_\ists  x &=
        \sum_{u\leadsto v}
        \weit_{uv} \kh{ \conj{x_u} \omega_k^{\cst_{uv}} x_v
        + \conj{x_v}\,\conj{\omega_k}^{\cst_{uv}} x_u } \notag \\
        &=\sum_{u\leadsto v}\weit_{uv}
        \kh{  \kh{ \conj{x_u} + \conj{x_v}\, \conj{\omega_k}\,^{\cst_{uv}}} \kh{ x_u + \omega_k^{\cst_{uv}} x_v }
    - \norm{x_u}^2 - \norm{x_v}^2 } \notag \\
    &= - \sum_{u\in V} d_u \norm{x_u}^2+
    \sum_{u\leadsto v} \weit_{uv} \norm{x_u + \omega_k^{\cst_{uv}} x_v}^2.
    \end{align}
 Combining these with $x^* D_\ists x = \sum\nolimits_{u\in V} d_u\norm{x_u}^2$ finishes
    the proof.
   \end{proof}

The lemma below presents a qualitative relationship between the eigenvector associated with $\lambda_1(\nL_\ists)$ and an assignment of $\ists$.

\begin{lemma}\label{lem:spectrum}
    All eigenvalues of $\nL_{\ists}$ are in the range $[0,2]$.
    Moreover, $\lambda_1(\nL_\ists) = 0$ if and only if
    there exists an assignment satisfying
    all equations in $\ists$.
\end{lemma}
\begin{proof}
To bound the eigenvalues of $\nL_{\ists}$, we look at the 
 following Rayleigh quotient
    \[
        \frac{x^* L_\ists x}{x^* D_\ists x},
\]
    where $x \neq 0$.
    By Lemma~\ref{lem:qL}, the numerator satisfies
    \[
        x^* L_\ists x = \sum_{u\leadsto v} \weit_{uv} \norm{x_u - \omega_k^{\cst_{uv}} x_v}^2 \geq 0
    \]
    and also
    \[
        x^* L_\ists x = 2\sum_{u\in V} d_u \norm{x_u}^2 - \sum_{u\leadsto v} \weit_{uv} \norm{x_u + \omega_k^{\cst_{uv}} x_v}^2
        \leq 2\sum_{u\in V} d_u \norm{x_u}^2 = 2 x^* D_\ists x.
 \]
    Therefore, the eigenvalues of $\calL_\ists$ lie in the range $[0,2]$.
    Moreover, $\lambda_1(\calL_\ists) = 0$ if and only if
    there exists an $x \in \mathbb{C}^n$
    such that $x^* L_\ists x = 0$, i.e.,
    \[
        \norm{x_u - \omega_k^{\cst_{uv}} x_v}^2 = 0
    \]
    holds for all $u\leadsto v$.
    The existence of such an $x$ is equivalent to the existence
    of an assignment satisfying all equations in $\ists$.
\end{proof}

\section{A Cheeger inequality for $\lambda_1(\nL_\ists)$ and \maxlin \label{sec:cheeger}}

The discrete Cheeger inequality~\cite{Alon86}
shows that, for any undirected graph $G$, the conductance $h_G$ of $G=(V,E)$ can be approximated by  the second smallest eigenvalue of $G$'s normalised Laplacian matrix $\mathcal{L}_G$, i.e.,
\begin{equation}\label{eq:cheeger}
\frac{\lambda_2(\mathcal{L}_G)}{2} \leq h_G\leq \sqrt{2\cdot\lambda_2(\mathcal{L}_G)}.
\end{equation}
Moreover, the proof of the second inequality above is constructive, and indicates that  a subset $S\subset V$ with conductance at most $\sqrt{2\cdot \lambda_2(\mathcal{L}_G)}$
can be found by using the second bottom eigenvector of $\mathcal{L}_G$ to embed vertices on the real line.
As one of the most fundamental results in spectral graph theory, the Cheeger inequality has found applications in the study of a wide range of optimisation problems, e.g., graph partitioning~\cite{journals/jacm/LeeGT14}, max-cut~\cite{Trevisan09}, and many practical problems like image segmentation~\cite{ShiM00} and web search~\cite{Kleinberg99}.

In this section, we develop connections between $\lambda_1(\calL_\ists)$
and \maxlin\ by proving a Cheeger-type inequality.  Let
\[
\asn: \setof{x_1,\ldots,x_n} \to [k]\union\setof{\una}
\]
be an arbitrary \emph{partial assignment} of an instance $\mathcal{I}$, where $\phi(x_i)=\una$ means that the assignment of $x_i$ has not been decided. These variables' assignments will be determined through some recursive construction, which will be elaborated in Section~\ref{sec:recursion}.   We remark that this framework of recursively computing a partial assignment was first 
introduced by Trevisan~\cite{Trevisan09}, and our theorem can be viewed as a generalisation of the one in \cite{Trevisan09},
 which corresponds to the $k = 2$ case of ours.

 To relate quadratic forms of $\mathcal{L}_\ists$  with the objective function of the \maxlin\ problem, we introduce a \emph{penalty} function as follows:

\begin{definition}\label{def:pnt} Given a partial assignment $\asn: \setof{x_1,\ldots,x_n} \to [k]\union\setof{\una}$ and a directed edge $(u,v)$, the penalty of  $(u,v)$ with respect to $\asn$ is defined by 
\begin{align}
    \pnt^\asn_{uv}(\ists) \defeq
    \begin{cases}
        0 & \asn(u)\neq \una,\asn(v)\neq \una, \asn(u) - \asn(v) \equiv \cst_{uv} \hmod k \\
        1 & \asn(u)\neq \una,\asn(v)\neq \una, \asn(u) - \asn(v) \not\equiv \cst_{uv} \hmod k \\
        0 & \asn(u) = \asn(v) = \una \\
        1 - 1/k \qquad & \text{exactly one of $\asn(u),\asn(v)$ is $\una$.}
    \end{cases}
\end{align}
For simplicity, we write $\pnt^\asn_{uv}$ when the underlying instance $\ists$ is clear from the context.
\end{definition}
The values of $ \pnt^\asn_{uv} $ from Definition~\ref{def:pnt} are chosen according to the following facts: (1) If both $u$ and $v$'s values are assigned, then their penalty is $1$ if the equation $\asn(u) - \asn(v) \not\equiv \cst_{uv} \hmod k$ associated with $(u,v)$ is unsatisfied, and $0$ otherwise; (2) If both $u$ and $v$'s values are $\una$, then their penalty is temporally set to $0$.  Their penalty will be computed when $u$ and $v$'s assignment are determined during a later recursive stage;
(3) If exactly one of $u,v$ is assigned, $\pnt^\asn_{uv}$ is set to $1 - 1/k$,
since a random assignment
to the other variable makes the edge $(u,v)$ satisfied with probability $1/k$.

Without loss of generality, we only consider $\phi$ for which $\phi(u)\ne \una$ for at least one vertex $u$, and  define the penalty of assignment $\asn$ by
\begin{align}
    \pnt^\asn \defeq \frac{2 \sum_{u\leadsto v} \weit_{uv} \pnt^\asn_{uv}}{\vol(\asn)},
\end{align}
where $\vol(\asn) \defeq \sum_{\asn(u)\neq \una} d_u$.
Notice that the $\pnt^\asn_{uv}$'s value is multiplied by $\weit_{uv}$ in accordance with the objective  of \maxlin\ which is to maximise the total weight of satisfied assignments. Also, we multiply $\pnt^\asn_{uv}$ by $2$ in the numerator 
since edges with at least one assigned endpoint are counted at most twice  in $\vol(\asn)$.
Notice that, as long as $G$ is weakly connected,
$\pnt^\asn=0$ if and only if all edges are satisfied by $\asn$ and, in general, the smaller the value of $\pnt^{\asn}$, the more edges are satisfied by $\asn$. With this in mind, we 
 define the \textit{imperfectness} $\pnt(\ists)$ of $\ists$ to quantify 
  how close $\ists$ is
to an instance where all equations can be satisfied by a single assignment.

\begin{definition}\label{def:perfect} Given any \maxlin\ instance $\ists = (G, k)$, the imperfectness of $\ists$ is defined by 
\begin{align}
    \pnt(\ists) \defeq \min_{\asn\in \kh{[k]\union\setof{\una}}^V\setminus \setof{\una}^V}\, \pnt^\asn.
\end{align}
\end{definition}

The main result of this section is a Cheeger-type inequality that relates 
$\pnt(\ists)$ and $\lambda_1(\nL_\ists)$, which is summarised in Theorem~\ref{thm:cheeger_k}.
Note that, since $\sin(x)\geq (2/\pi)\cdot x$  for $x\in [0, \pi/2]$,
 the factor before $\sqrt{2\lambda_1}$ in the theorem statement
is at most $(2 +k/4)$ for $k\geq 2$.

\begin{theorem}\label{thm:cheeger_k}
    Let $\lambda_1$ be the smallest eigenvalue of $\nL_\ists$.
    It holds that 
\begin{align}\label{eq:cheegerd2}
    \frac{\lambda_1}{2} \leq \pnt(\ists) \leq \kh{2 - \frac{2}{k} + \frac{1}{2\sin (\pi/k) }} \sqrt{2 \lambda_1}.
\end{align}
Moreover, given the eigenvector associated with $\lambda_1$,
there is an $O(m +n\log n)$-time algorithm  that returns
a partial assignment $\asn$
such that
\begin{align}\label{eq:cheegereq4}
    \frac{\lambda_1}{2} \leq \pnt^\asn \leq
    \kh{2 - \frac{2}{k} + \frac{1}{2\sin (\pi/k) }} \sqrt{2 \lambda_1}.
\end{align}
\end{theorem}

Our analysis is based on the following fact about the relations about the angle between two vectors and their Euclidean distance. For some $a,b\in \CCC$, we write $\theta(a,b)\in[-\pi,\pi)$ to denote the angle from $b$ to $a$, i.e.,  $\theta(a,b)$ is the unique real number in $[-\pi,\pi)$ such that
\[
    \frac{a}{\norm{a}} = \frac{b}{\norm{b}} \exp\left(i \theta(a,b)\right).
\]

\begin{fact}\label{fact:arc2}
   Let $a,b$ be complex numbers
   such that $\theta = \theta(a,b)$. The following statements hold: 
   \begin{enumerate}
       \item If $\theta\in \left[-\frac{2\pi}{k}, \frac{2\pi}{k} \right)$, then it holds that 
    \begin{align}\label{eq:arc2}
        \sizeof{\theta} \cdot \min\setof{\norm{a},\norm{b}} \leq \frac{\pi}{k\cdot\sin  (\pi/k)}\cdot \norm{a - b}.
    \end{align}
\item If $\theta  \in \left[-\pi, -\frac{2\pi}{k}\right) \union  \left[\frac{2\pi}{k}, \pi\right)$, then  it holds that 
    \begin{align}\label{eq:arc3}
        \min\setof{\norm{a},\norm{b}} \leq
        \frac{1}{2\cdot \sin(\pi/k) }\cdot \norm{a - b}.
    \end{align}
    \end{enumerate}
\end{fact}

\begin{proof}
    We assume $\theta\in \left[-\frac{2\pi}{k}, \frac{2\pi}{k} \right)$ and prove the first statement. 
Let 
\[
a' = \frac{a}{ \| a\|}\cdot\min \{ \|a\|, \|b\|\},
\]
\[
b' = \frac{b}{ \| b\|}\cdot\min \{ \|a\|, \|b\|\}.
\]
Then we have that 
     \begin{align}
         \norm{a - b} & \geq \|a'-b' \| \notag \\
         & \geq 2\cdot \sin\left(\frac{\sizeof{\theta}}{2}\right) \cdot \min\setof{\norm{a},\norm{b}} \notag \\
         & \geq 2\cdot \frac{\sizeof{\theta}}{2}\cdot \frac{\sin(\pi/k) }{\pi/k} \cdot \min\setof{\norm{a},\norm{b}} \notag \\
         & =    \sizeof{\theta}\cdot \frac{k \sin(\pi/k) }{\pi} \cdot \min\setof{\norm{a},\norm{b}},\nonumber
    \end{align}
    where the last inequality follows by the fact that
    for any $\alpha \in \left[0, \frac{\pi}{2}\right]$
    and $x \in [0, \alpha]$ it holds that
    $\sin x \geq x\cdot \frac{\sin \alpha}{\alpha}$.
    Multiplying $\frac{\pi}{k \sin\frac{\pi}{k}}$ on the both sides of the inequality above gives us (\ref{eq:arc2}).
    
    Now we prove the second statement. We have
    \begin{align}
        \norm{a - b} & \geq \norm{a'-b'} \notag \\
        & \geq 2\cdot \sin\left(\frac{\sizeof{\theta}}{2}\right)  \cdot \min\setof{\norm{a},\norm{b}} \notag \\
        & \geq 2\cdot \sin \left( \frac{\pi}{k} \right)\cdot \min\setof{\norm{a},\norm{b}}\nonumber
    \end{align}
    where the last inequality follows from the fact that
    $\theta  \in \left[-\pi, -\frac{2\pi}{k}\right) \union  \left[\frac{2\pi}{k}, \pi\right)$.
    Dividing both sides of the inequality above by $2\cdot \sin(\pi/k) $
    gives us~(\ref{eq:arc3}).
\end{proof}

\begin{proof}[Proof of Theorem~\ref{thm:cheeger_k}]
    We first prove $\frac{\lambda_1}{2}\leq \pnt(\ists)$.
    For a partial assignment $\asn: V\to [k]\union\setof{\una}$,
    we construct a vector $x_{\asn} \in \CCC^n$ by
    \begin{align}
        \kh{x_{\asn}}_u = 
        \begin{cases}
            \omega_k^{j} \quad & \asn(u) = j\in[k],\\
            0 & \asn(u) = \una.
        \end{cases}
    \end{align}
    Then, we have
    \begin{align}
        \pnt(\ists) = & \min_{\asn\in\kh{[k]\union\setof{\una}}^V \setminus \setof{\una}^V } \frac{2 \sum_{u\leadsto v} \weit_{uv}\pnt^\asn_{uv}}{\vol(\asn)} \notag \\
        \geq & \min_{\asn\in\kh{[k]\union\setof{\una}}^V\setminus \setof{\una}^V}         \frac{\sum_{u\leadsto v}
        \weit_{uv}\norm{\kh{x_\asn}_u - \omega_k^{\cst_{uv}} \kh{x_\asn}_v}^2}{2 \cdot \vol(\asn)} \notag \\ 
        = & \min_{\asn\in\kh{[k]\union\setof{\una}}^V \setminus \setof{\una}^V} \frac{x_{\asn}^* L_\ists x_{\asn}}
        {2 \cdot x_{\asn}^* D_\ists x_{\asn}} \notag \\
        \geq & \frac{1}{2}\cdot \min_{x\in \CCC^n, x\neq 0^n} \frac{x^* L_\ists x}{x^* D_\ists x} = \frac{\lambda_1(\nL_\ists)}{2},
    \end{align}
    where the second line follows from the fact that
    \begin{align}
        \norm{\kh{x_{\asn}}_u - \omega_k^{\cst_{uv}} \kh{x_{\asn}}_v}^2
        \leq 4 \cdot \pnt^\asn_{uv}
    \end{align}
    always holds for all $(u,v)\in E$,
    and the third line follows from
    Lemma~\ref{lem:qL} and that
    \begin{align}
        x_{\asn}^* D_\ists x_{\asn} =
        \sum_{u\in V} d_u \norm{\kh{x_{\asn}}_u}^2 =
        \sum_{u: \asn(u)\neq \una} d_u = \vol(\asn).
    \end{align}
This proves that $\lambda_1/2 \leq \pnt(\ists)$.

    Secondly,  we assume that $z\in\CCC^n$ is the vector such that
    \[
    \frac{z^* L_\ists z}{z^* D_\ists z} = \lambda_1,
    \]
    and prove the existence of an assignment $\phi$ based on $z$ satisfying 
    \[
    p^{\phi}\leq \left( 2- \frac{2}{k} + \frac{1}{2\sin(\pi/k)} \right) \sqrt{2\lambda_1},
    \]
which will imply (\ref{eq:cheegerd2}) and (\ref{eq:cheegereq4}). We scale each coordinate of $z$ and without loss of generality assume that $\max_{u\in V} \| z_u\| ^2=1$. 
For real numbers $t \geq 0$ and $\eta \in [0, \frac{2\pi}{k})$,
we define $k$ disjoint sets of vertices indexed by $j\in[k]$ as follows:
\begin{align}
    S^{(j)}_{t,\eta} =
    \setof{ u\ \left|\ \norm{z_u} \geq t\ \mathrm{and}\
\theta(z_u, \mathrm{e}^{i \eta }) \in \left[j \cdot \frac{2\pi}{k}, (j+1)\cdot \frac{2\pi}{k}\right) \right.}.
\end{align}
We then define an assignment $\asn_{t,\eta}$ where
\begin{align}
    \asn_{t,\eta}(u) =
    \begin{cases}
        j \quad & \exists j\in [k]: u\in S^{(j)}_{t,\eta}, \\
        \una \quad & \mathrm{otherwise.}
    \end{cases}
\end{align}
By definition, the $k$ vertex sets correspond to the vectors in the $k$ regions of the unit ball
after each vector is rotated by $\eta$ radians counterclockwise.
The role of $t$ is to only consider the coordinates $z_u$ with $\| z_u\|\geq t$.  This is illustrated in Figure~\ref{fig1}.

\definecolor{ashgrey}{rgb}{0.7, 0.75, 0.71}
	\definecolor{burlywood}{rgb}{0.87, 0.72, 0.53}
	\definecolor{cadetblue}{rgb}{0.37, 0.62, 0.63}
\definecolor{carolinablue}{rgb}{0.6, 0.73, 0.89}
		\definecolor{cinereous}{rgb}{0.6, 0.51, 0.48}
		\definecolor{coolblack}{rgb}{0.0, 0.18, 0.39}
	\definecolor{darkcerulean}{rgb}{0.03, 0.27, 0.49}
	\definecolor{dollarbill}{rgb}{0.52, 0.73, 0.4}
		\definecolor{graynew}{rgb}{0.75, 0.75, 0.75}
			\definecolor{grullo}{rgb}{0.66, 0.6, 0.53}
\begin{figure}[t]
\begin{center}
\begin{tikzpicture}[scale=2] 
 
 \coordinate (A) at (1.149,0.954);
 \coordinate (B) at (-1.410,0.513);
 \coordinate (C) at (0.260, -1.478);
 
 \draw[draw=white] (B) to[arc through ccw=(A)] (C) -- (arc through center) -- cycle;
\draw[draw=white] (A) to[arc through ccw=(C)] (B) -- (arc through center) -- cycle;
\draw[draw=burlywood] (C) to[arc through ccw=(B)] (A) -- (arc through center) -- cycle;

 \draw[ultra thick,draw=white] (0,0) -- (A);
 \draw[ultra thick,draw=white] (0,0) -- (C);
 
 \draw[thick,dashed] (0,0) -- (A);
 \draw[thick,dashed] (0,0) -- (B);
 \draw[thick,dashed] (0,0) -- (C);
  
  \draw[very thick, ->, draw=dollarbill,rotate=40] (1.5,0) arc(0:120:1.5cm and 1.5cm);
  \draw[ultra thick, ->, draw=darkcerulean, rotate=160 ] (1.5,0) arc(0:120:1.5cm and 1.5cm);
  \draw[ultra thick, ->, draw=orange, rotate=280 ] (1.5,0) arc(0:120:1.5cm and 1.5cm);

  \draw (0.6, 1.6) node[right] {{Set} $S^{(1)}_{t,\eta}$};
  
  \draw (-2.2, -0.85) node[right] {{Set} $S^{(2)}_{t,\eta}$};
  
  \draw (1, -1.3) node[right] {{Set} $S^{(3)}_{t,\eta}$};

  \draw[very thick, ->, draw=red, rotate=0] (0.7,0) arc(0:40:0.7cm and 0.7cm);
  
  \draw [fill=white] (0.82,0.1) rectangle (2.6,0.55);
  
  \draw (0.9, 0.45) node[right] {Random rotation};
  \draw (0.9, 0.23) node[right] {by $\eta\in[0, 2\pi/k)$};

 \draw[very thick, draw=black, ->] (-1.8,0) -- (1.8,0) coordinate (x axis);
 \draw[very thick,draw=black, ->] (0,-1.8) -- (0,1.82) coordinate (y axis);
 
 \draw[fill=cinereous!10] (0,0) circle (0.5cm);
 
 \draw[fill=grullo] (0.4,0) circle (0.02cm);

 \draw[fill=grullo] (0.8,0.8) circle (0.02cm);
 
 \draw[fill=grullo] (0.76,0.2) circle (0.02cm);
 
 \draw[fill=grullo] (-0.36,-0.2) circle (0.02cm);
 
 \draw[fill=grullo] (-0.36,-0.5) circle (0.02cm);

 \draw[fill=grullo] (-0.76,-1) circle (0.02cm);
 
 \draw[fill=grullo] (-0.66,-0.7) circle (0.02cm);
 
 \draw[fill=grullo] (-0.63534,-0.62) circle (0.02cm);
 
 \draw[fill=grullo] (-0.72,-0.734) circle (0.02cm);
 
 \draw[fill=grullo] (-0.7342,-0.6) circle (0.02cm);

 \draw[fill=grullo] (0.8,-0.45) circle (0.02cm);

 \draw[fill=grullo] (0.345,-0.3) circle (0.02cm);

 \draw[fill=grullo] (0.81,-0.7) circle (0.02cm);
 
 \draw[fill=grullo] (0.79,-0.64) circle (0.02cm);
 
 \draw[fill=grullo] (0.72,-1.2) circle (0.02cm);
 
 \draw[fill=grullo] (0.9,-1.1) circle (0.02cm);

 \draw[fill=grullo] (0.8,-0.81) circle (0.02cm);

 \draw[fill=grullo] (0.9,-0.7) circle (0.02cm);

 \draw[fill=grullo] (-0.6,0.85) circle (0.02cm);

 \draw[fill=grullo] (-0.8,0.75) circle (0.02cm);
 \draw[fill=grullo] (-0.76,0.72) circle (0.02cm);
 \draw[fill=grullo] (-0.7,0.85) circle (0.02cm);
 \draw[fill=grullo] (-0.5,1) circle (0.02cm);
 \draw[fill=grullo] (-0.8,0.9) circle (0.02cm);
  \draw[fill=grullo] (-1.1,0.6) circle (0.02cm);
 \draw[fill=grullo] (-1.2,0.75) circle (0.02cm);

 \draw[fill=grullo] (-0.3,0.25) circle (0.02cm);

 \draw[fill=grullo] (0.12,0.1) circle (0.02cm);

\draw[ultra thick,draw=black, ] (-0.07,0.5) -- (0.07,0.5) coordinate (y axis);

 \draw (-0.3, 0.6) node[right] {\textbf{\small $t$}};

  \end{tikzpicture}
 
\end{center}
\caption{Illustration of the proof for Theorem~\ref{thm:cheeger_k} for the case of $k=3$. The gray circle is obtained by sweeping  $t\in[0,1]$, and the  red arrow represents a random angle $\eta\in[0, 2\pi/k)$. A partial assignment is determined by the values of $\eta$ and $t$.\label{fig1}}
\end{figure}
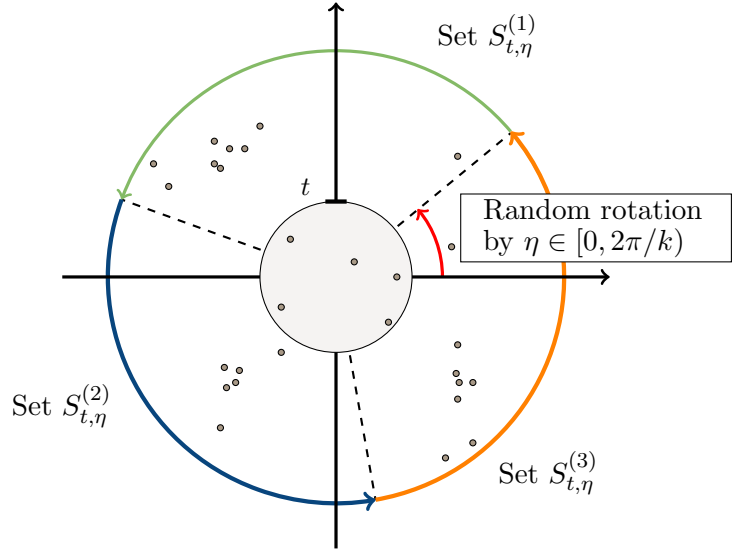

Our goal is to  construct probability distributions for $t$ and $\eta$ such that 
\begin{align}\label{eq:cheegermainobj4}
    \frac{\expec{t,\eta}{2\sum_{u\leadsto v} \weit_{uv} \pnt^\asn_{uv} }}{\expec{t,\eta}{\vol(\asn_{t,\eta})}}
        \leq
        \kh{2 - \frac{2}{k} +  \frac{1}{2\sin (\pi/k) }} \cdot \sqrt{2\cdot \frac{z^* L_\ists z}{z^* D_\ists z}}.
    \end{align}
This implies  by linearity of expectation that
\begin{align}
       \expec{t,\eta}{2\sum_{u\leadsto v} \weit_{uv} \pnt^\asn_{uv}  -
       \kh{2 - \frac{2}{k} +  \frac{1}{2\sin (\pi/k) }}\cdot \vol(\asn_{t,\eta})\cdot
   \sqrt{2\cdot \frac{z^* L_\ists z}{z^* D_\ists z}}} 
        \leq
        0,
    \end{align}
and existence of an assignment $\asn$ satisfying (\ref{eq:cheegereq4}).

Now let us assume that $t\in[0,1]$ is chosen such that $t^2$ follows from a uniform distribution over $[0,1]$, and $\eta$ is chosen uniformly at random from $[0, 2 \pi/k)$. We analyse the numerator and denominator in the left-hand side of (\ref{eq:cheegermainobj4}). 
For the denominator, it holds that 
\begin{align}\label{eq:x2}
    \expec{t,\eta}{\vol(\asn_{t,\eta})}  & =
    \sum_{u\in V} d_u\cdot \prob{}{  \asn(u) \neq \una} = \sum_{u\in V} d_u\cdot \prob{}{ \norm{z_u}\geq t} \notag \\
    & = \sum_{u\in V} d_u \norm{z_u}^2 = z^* D_\ists z.
\end{align}
For the numerator,
it holds by linearity of expectation that 
    \begin{align}
        \expec{t,\eta}{2\sum_{u\leadsto v} \weit_{uv} \pnt^\asn_{uv} }
        = 2 \sum_{u\leadsto v} \weit_{uv}\, \expec{t,\eta}{\pnt^\asn_{uv}}.
    \end{align}
    Then we look at $\expec{t,\eta}{\pnt^\asn_{uv}}$ for every edge $(u,v)\in E$.
    The analysis is based on the value of $\theta=\theta(z_u,  \omega_k^{\cst_{uv}} z_v)$,
    the angle from $z_v$ rotated by $2\cst_{uv}\pi/k$ radians clockwise
    to $z_u$.
\begin{itemize}
\item Case~1: 
    $\theta = \theta(z_u, \omega_k^{\cst_{uv}}z_v) \in \left[-\frac{2\pi}{k}, \frac{2\pi}{k}\right)$. It holds that 
    \begin{align}
        \expec{t,\eta}{\pnt^\asn_{uv}}
        =  & \kh{1 - \frac{1}{k}} \cdot \prob{}{ \| z_u\| < t\leq \|z_v\| \mbox{\ or\ } \| z_v\| < t\leq \|z_u\|} \notag \\
         & \qquad \qquad +1\cdot \prob{}{\norm{z_u}\geq t, \norm{z_v}\geq t, \asn(u) -\asn(v) \not\equiv \cst_{uv}\hmod k} \notag \\
        =  & \kh{1 - \frac{1}{k}} \sizeof{\norm{z_u}^2 - \norm{z_v}^2} +
        \frac{\sizeof{\theta}}{ 2\pi/k} \cdot \min\setof{\norm{z_u}^2, \norm{z_v}^2} \notag \\
        \leq & \kh{1 - \frac{1}{k}}\sizeof{\norm{z_u}^2 - \norm{z_v}^2} +
        \frac{{\pi}/\kh{k \sin(\pi/k)}}{2\pi/ k}\cdot \norm{z_u - \omega_k^{\cst_{uv}} z_v}
            \cdot \min\setof{\norm{z_u}, \norm{z_v}} \notag \\
            \leq & \kh{1 - \frac{1}{k}} \sizeof{\norm{z_u}^2 - \norm{z_v}^2} +
            \frac{1}{4\cdot\sin (\pi/k) } \cdot \norm{z_u - \omega_k^{\cst_{uv}} z_v} \cdot \kh{\norm{z_u} + \norm{z_v}}
        \nonumber \\
        \leq  & \kh{1 - \frac{1}{k} + \frac{1}{4\cdot \sin (\pi/k) }}
        \norm{z_u - \omega_k^{\cst_{uv}} z_v} \kh{\norm{z_u} + \norm{z_v} } \nonumber,
    \end{align}
where the second equality follows from that
    \begin{align}
        \lefteqn{\prob{}{\norm{z_u}\geq t, \norm{z_v}\geq t, \asn(u) -\asn(v) \not\equiv \cst_{uv}\hmod k}} \notag \\
        =  & \prob{}{\norm{z_u}\geq t, \norm{z_v}\geq t} \cdot 
    \prob{}{\asn(u) -\asn(v) \not\equiv \cst_{uv}\hmod k\ |\ \norm{z_u}\geq t, \norm{z_v}\geq t} \notag \\
    = & \min\setof{\norm{z_u}^2, \norm{z_v}^2}\cdot\frac{\sizeof{\theta}}{2\pi / k}, \notag
    \end{align}
    the third inequality follows by Fact~\ref{fact:arc2} and
    that $\sizeof{\theta}$ equals exactly the angle between $z_u$ and $\omega_k^{\cst_{uv}} z_v$.
         
\item Case~2:     $\theta = \theta(z_u, \omega_k^{\cst_{uv}}z_v) \in \left[-\pi, -\frac{2\pi}{k}\right)\union \left[\frac{2\pi}{k},\pi\right)$.
    It holds that
    \begin{align}
        \expec{t,\eta}{\pnt^\asn_{uv}}
        =  & \kh{1 - \frac{1}{k}} \cdot \prob{}{ \| z_u\| < t\leq \|z_v\| \mbox{\ or\ } \| z_v\| < t\leq \|z_u\|} \notag\\
         & \qquad\qquad +1\cdot \prob{}{\norm{z_u}\geq t, \norm{z_v}\geq t, \asn(u) -\asn(v) \not\equiv \cst_{uv}\hmod k} \notag \\
        = & \kh{1 - \frac{1}{k}}\cdot \sizeof{\norm{z_u}^2 - \norm{z_v}^2} +
        1\cdot \min\setof{\norm{z_u}^2, \norm{z_v}^2} \notag \\
        \leq & \kh{1 - \frac{1}{k}}\cdot \sizeof{\norm{z_u}^2 - \norm{z_v}^2} +
        \frac{1}{2\sin(\pi/k) }\cdot \norm{z_u - \omega_k^{\cst_{uv}} z_v}
            \cdot \min\setof{\norm{z_u}, \norm{z_v}} \notag \\
            \leq & \kh{1 - \frac{1}{k}}\cdot \sizeof{\norm{z_u}^2 - \norm{z_v}^2} +
            \frac{1}{4\sin (\pi/k) } \cdot \norm{z_u - \omega_k^{\cst_{uv}} z_v} \cdot \kh{\norm{z_u} + \norm{z_v}}
        \nonumber \\
        \leq & \kh{1 - \frac{1}{k} + \frac{1}{4\cdot\sin(\pi/k) }}
        \norm{z_u - \omega_k^{\cst_{uv}} z_v} \kh{\norm{z_u} + \norm{z_v} } \nonumber,
    \end{align}
    where
    the second equality follows from the fact that
    edge $(u,v)$ can not be satisfied when $\theta$ is in this range,
    the first inequality follows by Fact~\ref{fact:arc2}
    and that the angle between $z_u$ and $\omega_k^{\cst_{uv}}z_v$ is
    at least $\frac{2\pi}{k}$,
    and the last line follows by the triangle inequality.     
\end{itemize}

Combining these two cases gives us that 
    \begin{align}
        \expec{t,\eta} {2\sum_{u\leadsto v} \weit_{uv}\, {\pnt^\asn_{uv}}  }
        &\leq \kh{2 - \frac{2}{k} + \frac{1}{2\sin(\pi/k) }}
        \sum_{u\leadsto v} \weit_{uv} \norm{z_u - \omega_k^{\cst_{uv}} z_v}\kh{\norm{z_u} + \norm{z_v} } \notag \\
        &\leq \kh{2 - \frac{2}{k} + \frac{1}{2\sin(\pi/k) }}
        \sqrt{\sum_{u\leadsto v} \weit_{uv} \norm{z_u - \omega_k^{\cst_{uv}} z_v}^2}
        \sqrt{\sum_{u\leadsto v} \weit_{uv} \kh{\norm{z_u} + \norm{z_v} }^2}\notag \\
        &\leq \kh{2 - \frac{2}{k} + \frac{1}{2\sin(\pi/k) }}
        \sqrt{\sum_{u\leadsto v} \weit_{uv} \norm{z_u - \omega_k^{\cst_{uv}} z_v}^2}
        \sqrt{2\sum_{u} d_u \norm{z_u}^2 } \notag \\
        & = \kh{2 - \frac{2}{k} + \frac{1}{2\sin(\pi/k) }}
        \cdot \sqrt{z^* L_\ists z} \cdot \sqrt{2 z^* D_\ists z},
    \end{align}
  where
  the second inequality follows
  by the Cauchy-Schwarz inequality.
  Combining this with (\ref{eq:x2}) finishes the proof of the inequality (\ref{eq:cheegerd2}).

  Finally, let us look at the time needed to
  find the desired partial assignment. Notice that, by the law of total expectation, we can write   \begin{align}
      & \expec{t,\eta}{2\sum_{u\leadsto v} \weit_{uv} \pnt^\asn_{uv}  -
       \kh{2 - \frac{2}{k} +  \frac{1}{2\sin (\pi/k) }}\cdot \vol(\asn_{t,\eta})\cdot
   \sqrt{2\cdot \frac{z^* L_\ists z}{z^* D_\ists z}}} \notag\\
       = & \expec{t_0}{\expec{t,\eta}{\left. 2\sum_{u\leadsto v} \weit_{uv} \pnt^\asn_{uv}  -
       \kh{2 - \frac{2}{k} +  \frac{1}{2\sin (\pi/k) }}\cdot \vol(\asn_{t,\eta})\cdot
   \sqrt{2\cdot \frac{z^* L_\ists z}{z^* D_\ists z}}\ \right|\ t = t_0}}.
  \end{align}
   As a preparation step, we build 
    two ordered sequences of coordinates of $\{ z_u\}$:
  the first ordered sequence is based on $z_u$'s norm, and the other ordered sequence is based on $z_u$'s angle. This step takes $O(n\log n)$ time.
Now we construct two sequences of sweep sets:  the first is based on  $t$, and the second is based on $\eta$. For constructing the sweep sets based on $t$, the algorithm increases $t$ from $0$ to $1$,
    and updates the conditional expectation of the edges incident with $u$ whenever $t$ exceeds $\norm{z_u}$. Notice that each edge $(u,v)$ will be updated at most twice, i.e., in the step when $t$ reaches $\norm{z_u}$  and when it reaches $\norm{z_v}$. 
    Hence, the total runtime for constructing the sweep sets  on $t$ is $O(m)$. The runtime analysis for constructing the sweep sets  on $\eta$ is similar: the algorithm increases $\eta$ from $0$ to $2\pi/k$, and updates the penalties $\pnt^\asn_{uv}$ of the edges $(u,v)$ only if the assignment of $u$ of $v$ changes.  Since every edge will be updated at most twice, the total runtime for constructing the sweep sets on $\eta$ is $O(m)$ as well. The algorithm terminates if the 
        assignment $\asn$ satisfying~(\ref{eq:cheegereq4}) is found. The total runtime of the algorithm is $O(m+n\log n)$.
  \end{proof}

\begin{remark}
    We remark that the factors $\lambda_1/2$ and $\sqrt{\lambda_1}$ in
    Theorem~\ref{thm:cheeger_k} are both tight
    within constant factors.
    The tightness can be derived directly from Section 5 of~\cite{Trevisan09},
    since when $k=2$, our inequality is the same as the one in~\cite{Trevisan09} up to constant factors.

    We also remark that the factor of $k$ in Theorem~\ref{thm:cheeger_k} is necessary, which is shown by the following instance: the linear system has $nk$ variables where every variable belongs to one of $k$ sets $S_0,\dots,S_{k-1}$ with $|S_i| = n$ for any $0 \le i \le k-1$. Now, for any $i$, we add $n$ equations of the form $x_u - x_v = 1 \mod k$ with $x_u \in S_i$, $x_v \in S_j$, and $j = (i+1) \mod k$, and $n$ equations of the form $x_u - x_v = 1 \mod k$ with $x_u \in S_i$, $x_v \in S_j$, and $j = (i+2) \mod k$. This instance is constructed such that the underlying graph is regular, and every assignment could only satisfy at most half of the equations, implying that the imperfectness is $p(\ists) = \Omega(1)$. However, mapping each variable in $S_i$ to the root of unity $\omega_k^i$, it's easy to see that $\lambda_1(\calL_{\ists}) = O(1/k^2)$. Hence Theorem~\ref{thm:cheeger_k} is tight with respect to $k$.  
\end{remark}

\begin{remark}
  We  notice that  this factor of $k$ originates from the relation between the quadratic forms of the Hermitian Laplacian and the penalty function $p$ of Definition~\ref{def:pnt}. Indeed, we could re-define our penalty function such that, for an equation of the form $x_u - x_v = c \mod k$ and assignment  $\phi(u) -\phi(v) = d \neq c \mod k$, the value of the penalty to this equation with respect to $\phi$ is proportional to $\min\setof{|c-d|,k-|c-d|}$, i.e., the distance between $c$ and $d$. Based on this new penalty function, we could obtain the same Cheeger inequality without any dependency on $k$. However, with this new penalty function we would end up solving a different version of the original \maxlin\ problem. 
\end{remark}

Finally, we  compare the proof techniques of Theorem~\ref{thm:cheeger_k} with other Cheeger-type inequalities in the literature: first of all, most of the Cheeger-type inequalities~(e.g.,~\cite{Alon86,Trevisan09,journals/jacm/LeeGT14,KwokLLGT13}) consider the case where every eigenvector is in $\RRR^n$ and are only applicable for undirected graphs, while for our problem the graph $G$ associated with $\ists$ is directed and eigenvectors of $\mathcal{L}_{\ists}$ are in $\CCC^n$. Therefore, constructing sweep sets in $\CCC$ is needed, which is more involved than proving similar Cheeger-type inequalities~(e.g.,~\cite{Alon86, Trevisan09}). Secondly, 
by dividing the complex unit ball into $k$ regions, we 
are able to show that a partial assignment corresponding to $k$ disjoint subsets can be found using a single eigenvector. This is quite different from the  
techniques used for finding $k$ vertex-disjoint subsets of low conductance in an undirected graph, where $k$ eigenvectors are usually needed~(e.g.~\cite{journals/jacm/LeeGT14,KwokLLGT13, PSZ17}).

It is also worth mentioning that a Cheeger-like inequality was shown in \cite{BandeiraSS13} for a synchronisation problem which has some connections to \maxlin.~Their analysis, however, cannot be adapted in our setting.
We also remark that, while sweeping through values of $t$ is needed to obtain \emph{any} guarantee on the penalty of the partial assignment computed, we could in principle just choose a random angle $\eta$: in this way, however, the partial assignment returned would satisfy (\ref{eq:cheegereq4}) only in expectation.

\section{Sparsification for \maxlin \label{sec:sparsification}}
We have seen in Section~\ref{sec:cheeger} that, given any vector in $\CCC^n$ whose quadratic form in $\calL_{\ists}$ is close to $\lambda_1(\calL_\ists)$, we can compute a partial assignment of $\ists$ with bounded approximation guarantee. In Section~\ref{sec:recursion} we will show that  a total assignment can be found by recursively applying this procedure on variables for which an assignment has not yet been fixed. In particular, we will show  that every iteration takes a time nearly-linear in the number of equations of our instance, which can be quadratic in the number of variables. To speed-up each iteration and obtain a time per iteration that is nearly-linear in the number of variables, we need to  sparsify our input instance $\ists$. 

In this section we show that the construction of spectral sparsifiers by 
effective resistance sampling
introduced by Spielman and Srivastava~\cite{SpielmanS11} can be generalised to sparsify \maxlin\ instances. 
In particular,
given an instance $\ists$ of \maxlin\ with $n$ variables and $m$ equations,
we can find in nearly-linear time
a sparsified instance $\jsts$ with $O(nk\log(nk))$ equations
such that for any partial assignment $\asn: V\to [k]\union\setof{\una}$,
the number of unsatisfied equations in $\jsts$ is preserved within a constant factor.
This means that we can apply our algorithm for \maxlin\ to a sparsified instance $\jsts$,
and  any dependency on $m$ in our runtime can be replaced by $nk\log(nk)$.
We remark that we could simply apply uniform sampling to obtain a sparsified instance. However, this would in the end result in an additive error in the fraction of unsatisfied equations, much like in the case of the original Trevisan's result for \textsf{MAX-CUT}~\cite{Trevisan09}. With our construction, instead, we only lose a small multiplicative error. For completeness of discussion, we first recall the definition of a spectral sparsifier.

\begin{definition}
Let $G=(V,E,w)$  be an arbitrary undirected graph with $n$ vertices and $m$ edges. We call a sparse subgraph $H$ of $G$, with proper reweighting of the edges,  a $(1+\delta)$-spectral sparsifier of $G$ if
\[
(1-\delta)x^{\rot}L_Gx\leq x^{\rot} L_{H} x\leq (1+\delta) x^{\rot} L_Gx
\]
holds for any $x\in\RRR^n$,
where $L_G$ and $L_{H}$ are the respective Laplacian matrices of $G$ and $H$. 
\end{definition}

To  construct a sparsified instance  $\jsts$, we introduce label-extended graphs and their Laplacian matrices to characterise the original \maxlin\ instance.
Let $P \in \mathbb{R}^{k\times k}$ be the permutation matrix where 
$P_{ij}=1$ if $i \equiv j + 1 \hmod k$, and $P_{ij}=0$ otherwise.
We  define the adjacency matrix $\AP_{\ists} \in \kh{\mathbb{R}^{k\times k}}^{n\times n}$
for the label-extended graph of  instance $\ists$,
where each entry of $\AP_{\ists}$ is a  matrix in $\mathbb{R}^{k\times k}$ given by
\begin{align}
    (\AP_\ists)_{uv} \defeq
    \begin{cases}
        \weit_{uv}P^{\cst_{uv}} & u\leadsto v, \\
        \weit_{vu}\kh{P^{\rot}}^{\cst_{vu}} & v\leadsto u, \\
        0 & \mathrm{otherwise}.
    \end{cases}
\end{align}
We then define the degree-diagonal matrix $\DP_\ists\in\kh{\mathbb{R}^{k\times k}}^{n\times n}$
by $(\DP_\ists)_{uu} = d_u\cdot I_{k\times k}$,
where $I_{k\times k}$ is the $k\times k$ identity matrix,
and define the Laplacian matrix by 
\begin{align}
    & \LP_{\ists} = \DP_{\ists} - \AP_{\ists}. 
\end{align}
Notice that the Hermitian Laplacian $L_{\ists}$ is a \emph{compression} of $\LP_{\ists}$, i.e., there exists an orthogonal projection $U$ such that $U^*\LP_{\ists}U = L_{\ists}$.

We further  write 
 $\LP_{\ists} = \DP_{\ists} - \AP_{\ists}$ as a sum of matrices, each one
corresponding to a single equation.
More precisely, for an equation $u_i - v_i \equiv \cst_{i} \hmod k$ with weight $\weit_{u_iv_i}$,
we define a matrix $\bp_{uv}\in\kh{\mathbb{R}^{k\times k}}^{n\times 1}$ by
\begin{align}
    \kh{\bp_{uv}}_{w} =
    \begin{cases}
        I_{k\times k}& w = u, \\
        -\kh{P^{\rot}}^{\cst_{uv}} & w = v, \\
        0 & \mathrm{otherwise.}
    \end{cases}
\end{align}
Then it is easy to verify that
\begin{equation}\label{eq:llll}
    \LP_\ists = \sum_{u\leadsto v} \weit_{uv} \bp_{uv} \bp_{uv}^{\rot},
    \end{equation}
    and it holds for any  $x \in \kh{\mathbb{R}^{k}}^n$ that
\begin{align}\label{eq:qua}
    x^{\rot} \LP_{\ists} x = \sum_{u\leadsto v} \weit_{uv} \norm{x_u - P^{\cst_{uv}} x_v}^2.
\end{align}
For any assignment $\asn : V \to [k]$,
we construct an indicator vector $\xp_\ists\in\kh{\mathbb{R}^k}^n$ by
$
    \kh{\xp_\ists}_u = e_{\asn(u) + 1}
$,
where $e_j \in \mathbb{R}^k$ is the $j$-th standard basis vector. 
Then  it is easy to see that the total weight of unsatisfied equations for $\asn$ is  $(1/2)\cdot \xp_\ists^{\rot} \LP_\ists \xp_\ists$\footnote{
We remark that, if we use the Hermitian Laplacian matrices $L_\ists$ directly instead, this relation only holds up to an $O(k)$ factor. That is why we sparsify the matrix $\LP_{\ists}$ instead. }.

Next we will present an algorithm that produces a sparse \maxlin\ instance $\jsts$ from $\ists$ such that the total weight of unsatisfied equations is preserved\footnote{Notice that we can decide whether there is an assignment satisfying all the equations in $\ists$  by fixing the assignment of an arbitrary vertex and determining assignments for other vertices accordingly, and therefore we only need to consider the case when $\ists$ is unsatisfiable.}. Our algorithm can be described as follows:
first we sample every edge $(u,v)$ in $\ists$ with a certain probability $ p_{uv}$,  and set the weight of every sampled edge $(u,v)$ as its original weight multiplied by $1/ p_{u,v}$. Then, we output an instance $\jsts$ which consists of all the sampled edges.  Notice that this sampling scheme ensures that $\E[\LP_\jsts] = \LP_\ists$,  but we need to choose $p_{uv}$ properly to ensure that (1) $\LP_\jsts$ is sparse,  and (2) $\LP_\jsts$
approximates $\LP_\ists$ with high probability. 
We remark that, while our algorithm and analysis closely follow the one by Spielman and Srivastava~\cite{SpielmanS11}, the requirement of our output is slightly stronger:  in addition to the sparsity constraint for $\jsts$, we need to ensure that the output $\jsts$ is a valid \maxlin\ instance.

To analyse the algorithm, for every edge $(u,v)$ let  $X_{uv}$ be a random matrix defined by 
\begin{align}
    X_{uv} =
    \begin{cases}
        \frac{\weit_{uv}}{\pb_{uv}}\cdot \LP_\ists^{-1/2}\bp_{uv} \bp_{uv}^{\rot} \LP_\ists^{-1/2}&
        \text{with probability $\pb_{uv}$,} \\
        0  & \text{with probability $1 - \pb_{uv}$.}
    \end{cases}
\end{align}
We set the probabilities to
\begin{align}
    \pb_{uv} = \min\setof{1, 10 \left(1/\delta^2\right) \log(nk)  \lev_{uv} }
\end{align}
where $\lev_{uv}$ is defined by
\begin{align}
    \lev_{uv} = \weit_{uv} \trace{\LP_\ists^{-1/2} \bp_{uv} \bp_{uv}^{\rot} \LP_\ists^{-1/2}}.
\end{align}
Notice that by the definition of $\lev_{uv}$ we have that 
\begin{equation}\label{eq:boundluv}
 \sum_{u\leadsto v} \lev_{uv}  = \sum_{u\leadsto v}\weit_{uv}\cdot  \trace{\LP_\ists^{-1/2} \bp_{uv} \bp_{uv}^{\rot} \LP_\ists^{-1/2}} = \sum_{u\leadsto v}\weit_{uv}\cdot  \trace{\LP_\ists^{-1} \bp_{uv} \bp_{uv}^{\rot}}  =nk.
\end{equation}
We also assume without loss of generality that
 $\pb_{uv}< 1$ holds for all edges $(u,v)$. Otherwise,  we split every edge $(u,v)$ with 
 \begin{align}
    \lev_{uv} \geq \frac{1}{10\cdot  (1/\delta^2)\cdot \log(nk)}
\end{align}
 into $K = \ceil{10\cdot  (1/\delta^2)\cdot \log(nk) }$ parallel edges,
each of which has  weight $\weit_{uv} / K$.
By (\ref{eq:boundluv}) there are at most $O\left( (1/\delta^2)\cdot nk\log(nk) \right)$ such edges.

The following 
 matrix Chernoff bound will be used in our analysis.

\begin{lemma}[\cite{Tro12}]\label{lem:chernoff}
    Let $X_1,\ldots,X_m$ be independent random $n$-dimensional symmetric positive semidefinite matrices
    such that
    \begin{itemize}
        \item $\expec{}{X} = I$ where $X = \sum_{i=1}^m X_i$ and
        \item $\norm{X_i} \leq R$ holds for all $i = 1,\ldots,m$.
    \end{itemize}
    Then, for any $0< \delta < 1$,
    \begin{align}
        & \prob{}{\lambdamin\kh{\sum\nolimits_{i=1}^m X_i} \leq (1 - \delta)}
        \leq n \cdot \exp\kh{- \frac{\delta^2}{2R}}, \notag \\
        & \prob{}{\lambdamax\kh{\sum\nolimits_{i=1}^m X_i} \geq (1 + \delta)}
        \leq n \cdot \exp\kh{- \frac{\delta^2}{3R}}.
    \end{align}
\end{lemma}

\begin{theorem}\label{thm:sparsification}
    There is an algorithm that,
    given an unsatisfiable instance $\ists$ of \maxlin\ with
    $n$ variables and $m$ equations
    and parameter $0 < \delta < 1$,
    returns in $\widetilde{O}(mk)$ time 
    an instance $\jsts$ with
    the same set of variables and    $O\left((1/\delta^2)\cdot nk \log(nk)\right)$ equations.
    Furthermore,
     with high probability
     it holds for any vector $x\in\kh{\mathbb{R}^{k}}^n$ that
    \begin{align}\label{eq:gua}
        (1 - \delta) x^{\rot} \LP_\ists x \leq
        x^{\rot} \LP_\jsts x \leq
        (1 + \delta) x^{\rot} \LP_\ists x.
    \end{align}
\end{theorem}

\begin{proof}[Proof of Theorem~\ref{thm:sparsification}]
We first prove the approximation guarantee~(\ref{eq:gua}). Since it holds that 
    \begin{align}
        \expec{}{\sum_{u\leadsto v} X_{uv}}
        = \sum_{u\leadsto v}\expec{}{X_{uv}}
        = \sum_{u\leadsto v} \weit_{uv}\LP_\ists^{-1/2}\bp_{uv} \bp_{uv}^{\rot} \LP_\ists^{-1/2}
        = I,
    \end{align}and
    \begin{align}
        \norm{X_{uv}} \leq & \frac{\delta^2}{10\log(nk)}
        \cdot
        \frac{\norm{\LP_\ists^{-1/2} \bp_{uv} \bp_{uv}^T \LP_\ists^{-1/2}}}
        {\trace{\LP_\ists^{-1/2} \bp_{uv} \bp_{uv}^T \LP_\ists^{-1/2}}}
        \leq 
        \frac{\delta^2}{10\log(nk)}.
    \end{align}
    By applying Lemma~\ref{lem:chernoff}, the approximation guarantee holds.    
        The number of edges in $\jsts$ follows from (\ref{eq:boundluv}) and Markov's inequality.

It remains to analyse the runtime needed to compute $\lev_{uv}$ for all edges. 
    To this end, we will need the Johnson-Lindenstrauss lemma~\cite{JL84,Ach03}
    and nearly-linear time Laplacian solvers~\cite{ST14,CKM+14}.
    Specifically, we  write $\lev_{uv}$ as
    \begin{align}
        \lev_{uv} = & \weit_{uv}\cdot \trace{\LP_\ists^{-1/2} \bp_{uv} \bp_{uv}^{\rot}\LP_\ists^{-1/2}} \notag \\
        = & \weit_{uv}\cdot  \trace{\bp_{uv}^{\rot} \LP_\ists^{-1} \bp_{uv}} \notag \\
        = & \weit_{uv}\cdot  \trace{\bp_{uv}^{\rot} \LP_\ists^{-1} \kh{\sum_{x\leadsto y} \weit_{xy} \bp_{xy}\bp_{xy}^{\rot}} \LP_\ists^{-1} \bp_{uv}}
        \notag \\
        = & \weit_{uv} \cdot \trace{\bp_{uv}^{\rot} \LP_\ists^{-1} \fp^{\rot} \fp \LP_\ists^{-1} \bp_{uv}} \notag \\
        = & \weit_{uv} \cdot \norm{\fp \LP_\ists^{-1} \bp_{uv}}_\mathrm{F}^2,
    \end{align}
    where the second equality follows from the cyclicality of trace,
    and in the last two lines we write $\fp$ to denote the matrix whose rows
    are $\sqrt{\weit_{xy}}\ \bp_{xy}^{\rot}$.
    Now we generate a matrix $Q$ of size $q\times mk$ with random $\pm 1/\sqrt{q}$ entries,
    where $q = 100 \log(nk)$.
    By the Johnson-Lindenstrauss lemma,
    with high probability it holds for all edges  $(u,v)$ that 
    \begin{align}
        \frac{1}{2} \norm{\fp \LP_\ists^{-1} \bp_{uv}}_\mathrm{F}^2 \leq
        \norm{Q \fp \LP_\ists^{-1} \bp_{uv}}_\mathrm{F}^2 \leq
        \frac{3}{2} \norm{\fp \LP_\ists^{-1} \bp_{uv}}_\mathrm{F}^2.
    \end{align}
    Therefore, our runtime  follows by computing 
    every row of  $Q \fp \LP_\ists^{-1} \bp_{uv}$ 
    by a nearly-linear time Laplacian solver.
\end{proof}


\section{Algorithm for \maxlin\label{sec:recursion}}

\newcommand{\RecursivePartition}{\textsc{RecursiveConstruct}}

Theorem~\ref{thm:sparsification} tells us  that, given an instance $\ists^*$,  we can find a sparse instance $\ists$ so that the quadratic forms of the corresponding Laplacians $\calL_{\ists^*}$ and $\calL_\ists$ are related by (\ref{eq:gua}). Therefore throughout this section we assume that the input instance $\ists$ for \maxlin\ with $n$ variables has $m = \widetilde{O}\left((1 / \delta^2)\cdot nk\right)$ equations for some parameter $\delta>0$.
Recall that Theorem~\ref{thm:cheeger_k} shows that, for any \maxlin\ instance $\ists$, given an eigenvector for the smallest eigenvalue $\lambda_1(\calL_\ists)$, we can obtain a partial assignment  $\phi$ satisfying 
\begin{align}
\label{eq:cheegerprecise}
        \pnt^\asn \leq
        \kh{2 - \frac{2}{k} + \frac{1}{2\sin (\pi/k) }} \sqrt{2 \lambda_1}.
\end{align}
    Now we show that, by a repeated application of Theorem~\ref{thm:cheeger_k} on the subset of the equations of $\ists$ for which both variables are unassigned, we can obtain a full assignment of $\ists$.  Our algorithm closely follows the one by Trevisan~\cite{Trevisan09} and is described in Algorithm~\ref{algo1}.

 \begin{algorithm}
  \caption{ $\RecursivePartition(\ists,\delta)$\label{algo1}}
  \label{influx}
  \begin{algorithmic}[1]
   \State Compute  vector $z\in \CCC^n$ satisfying
    \begin{align}\label{eq:cond}
        \frac{z^* L_\ists z}{z^* D_\ists z} \leq (1 + 2 \delta) \lambda_1(\calL_\ists);
    \end{align}   
    
   \State Apply the algorithm from Theorem~\ref{thm:cheeger_k} to 
    compute  $\phi: V \to [k]\union\setof{\una}$ such that
    \begin{align}
        \pnt^\asn \leq
        (1 + \delta)
        \kh{2 - \frac{2}{k} + \frac{1}{2\sin (\pi/k) }} \sqrt{2 \lambda_1};
    \end{align}
    \If{$2 \pnt^\asn \geq \kh{1 - 1/k} \vol(\phi)$} 
    \State return random full assignment $\phi': V\to [k]$; \\
    \Comment{the case where the current assignment is worse than a random assignment}
    \ElsIf{$\asn$ \rm{is a full assignment (i.e.} $\asn(V) \subseteq [k]$\rm{)}} 
    \State return $\asn$; \\
    \Comment{The recursion terminates if every variable's assignment is determined}
    \Else
     \State $\ists' \gets$ set of equations from   $\ists$ in which both variables' assignments are not determined; 
     \If{$\ists'=\emptyset$} 
     \State set $\phi(u)$ to be an arbitrary assignment if $\phi(u) =\una$ for any $u$;
     \State return $\phi$;
     \Else
     \State $\asn_1 \gets \RecursivePartition(\ists',\delta)$;
     \State return $\asn \union \asn_1$; 
     \EndIf
    \EndIf
    \end{algorithmic}
\end{algorithm}

To achieve the guarantees of (\ref{eq:cheegerprecise}), however, we would need to compute the eigenvector corresponding to $\lambda_1(\calL_\ists)$ \emph{exactly}. To obtain a nearly-linear time algorithm, instead, we relax this requirement and  compute a vector $z$ that well-approximates this eigenvector. In particular, the following lemma shows that, for any $\delta$, we can compute a vector $z \in \mathbb{C}^n$ 
   satisfying (\ref{eq:cond}) in nearly-linear time. 
 
\begin{lemma}
\label{lem:powermethod}
For any given error parameter $\delta$, there is an $\tilde{O}\left( \left(1/\delta^{3}\right) \cdot kn\right)$ time algorithm that returns $z\in\CCC^n$ satisfying (\ref{eq:cond}).
\end{lemma}    

\begin{proof} Following the discussion in 
 \cite[Section 8.2]{vishnoilx}, we  compute a vector $z\in\CCC^n$ satisfying (\ref{eq:cond}) in $O\left((1/\delta)\cdot \log({n/\delta})\right)$ iterations by the power method, where each iteration consists in solving a linear system of the form $\calL_\ists x = b$  for some vector $b$. This can be done up to $\delta$ precision in  $O\left(\left(m + n \log^2n\right)\log(1/\delta)\right)$-time using a nearly-linear time solver for connection Laplacians~\cite{kyng16,KyngS16}. The total running time follows from our assumption on $m$.
\end{proof}

To analyse Algorithm~\ref{algo1}, we introduce some notation. Let $t$ be the number of recursive executions of  Algorithm~\ref{algo1}. For any $1\leq j \leq t+1$, let $\ists_{j}$ be the instance of \maxlin\ in the  $j$-th execution. We indicate with $\rho_j m$ the number of equations in $\ists_j$, where $0\leq \rho_j\leq 1$. Notice that  $\ists_{1} = \ists$ and $\ists_{t+1} = \emptyset$. 
We assume that the maximum number of equations in $\ists_{j}$ that can be satisfied by an assignment is $(1-\varepsilon_j) \rho_j m$, with $\varepsilon=\varepsilon_1$. Also notice that it holds for any $1\leq j\leq t$ that $\eps_j \rho_j m \leq \eps m$,
which implies 
\begin{equation}\label{eq:eq11}
\eps_j \leq \eps / \rho_j.
\end{equation} 
The next theorem presents the performance of our algorithm, whose informal version is Theorem~\ref{thm:main1}

\begin{theorem}\label{thm:algoanalysis}
    Given an instance $\ists$ of \maxlin\ 
    whose optimum is $1 - \eps$
    and a parameter $\delta > 0$,
    the algorithm $\RecursivePartition(\ists,\delta)$
    returns in $\widetilde{O}\left(  \left(1/\delta^{3}\right)\cdot k n^2\right)$ time
    an assignment $\asn$ satisfying at least 
    $1 - 8 \nu \sqrt{\eps}$ fraction of the equations, where
    \[
        \kpa \defeq (1 + \delta)
        \kh{2 - \frac{2}{k} + \frac{1}{2\sin (\pi/k) }} = O(k).
    \]
\end{theorem}
\begin{proof}
    Suppose we are now at the $j$-th iteration.
    By  Theorem~\ref{thm:cheeger_k}, we know that
    the total weight of unsatisfied equations in $\ists_j \setminus \ists_{j+1}$
    is at most
    \begin{align*}
    \lefteqn{2\cdot (\rho_j - \rho_{j+1})  m \left( 2 - \frac{2}{k} + \frac{1}{2\sin(\pi/k)}\right)\cdot \sqrt{2\cdot (1+2\delta) \cdot\lambda_1\left(\mathcal{L}_{\ists_j} \right) }}\\
    & \leq 2\cdot(\rho_j - \rho_{j+1})  m (1+\delta)  \left( 2 - \frac{2}{k} + \frac{1}{2\sin(\pi/k)}\right) \sqrt{2\cdot\lambda_1\left(\mathcal{L}_{\ists_j} \right) } \\
    & = 2 \cdot(\rho_j - \rho_{j+1})  m \kpa \sqrt{2\cdot\lambda_1\left(L_{\ists_j} \right) } \\
    & \leq 4 \cdot(\rho_j - \rho_{j+1})  m \kpa \sqrt{\varepsilon_j }\\
    & \leq 4 \cdot(\rho_j - \rho_{j+1})  m \kpa \sqrt{\varepsilon/\rho_j }\\
    & \leq  4\cdot  m \kpa \sqrt{\varepsilon} \int_{\rho_{j+1}}^{\rho_j}
        \sqrt{\frac{1}{r}} \ \mathrm{d} r,
    \end{align*}
      Therefore, the total weight of unsatisfied equations in $\ists$
    can be upper bounded by 
    \[
         4  m \kpa \sqrt{\varepsilon}\ \sum_{j=1}^{t}
        \int_{\rho_{j+1}}^{\rho_j}
        \sqrt{\frac{1}{r}} \mathrm{d} r
        \leq
         4  m \kpa \sqrt{\varepsilon}
        \int_{0}^{1}
        \sqrt{\frac{1}{r}}\ \mathrm{d} r
        = 8m \kpa \sqrt{\varepsilon},
    \]
    which implies that the total weight of satisfied equations
    is at least $(1 - 8\kpa \sqrt{\eps})m$. The runtime follows by Lemma~\ref{lem:powermethod} and the fact that we perform at most a linear number of recursive iterations.
\end{proof}

The following corollary
which states how much our algorithm beats a random assignment
follows from Theorem~\ref{thm:main1}.

\begin{cor}\label{cor:k3}
    Given a   \maxlin\  instance $\ists$
    whose optimum is $\xi$
    and a constant $\delta > 0$,
    Algorithm~\ref{algo1}
    returns in $\tilde{O}\left( \delta^{-3} n^2\right)$ time
    an assignment $\asn$ satisfying at least 
    $\kh{1/k + \tau}\xi$ fraction of the equations, where
$
        \tau = \Omega\kh{\frac{1}{k^3}}$.
\end{cor}

\begin{proof}
    We define the parameter \[
    \eps' = \frac{(1 - \frac{1}{k})^2}{64\kpa^2},\]
     which implies that $1 - 8\kpa \sqrt{\eps'} = 1/k$.
    Since Algorithm~\ref{algo1} always chooses the best between 
    the assignment found by recursively applications of Theorem~\ref{thm:cheeger_k} and a random assignment, the algorithm's approximation ratio is at least 
    \begin{align*}
        & \frac{\max\setof{1 - 8\kpa \sqrt{\eps}, 1/k}}{1 - \eps} 
        \geq \frac{1/k}{1 - \eps'} 
        \geq \frac{1/k}{1 - \frac{1}{256\kpa^2}} 
        \geq \frac{1/k}{1 - \frac{1}{512(1 + \delta)^2 k^2}} 
        \geq \frac{1}{k} + \frac{1}{(1 + \delta)^2 k^3} 
        = \frac{1}{k} + \Omega\kh{\frac{1}{k^3}},
    \end{align*}
  where the first inequality follows by the fact that
    $\frac{1 - 8\kpa\sqrt{\eps}}{1-\eps}$ is a monotone decreasing function in $\varepsilon$,
    the third inequality follows from the definition of $\kpa$,
    and the last inequality follows from that $\delta$ is a constant. \end{proof}

\section{Algorithm for \maxlin\ on expanders \label{sec:expander}}

In this section we further develop techniques for analysing Hermitian Laplacian matrices by presenting a
subquadratic-time approximation algorithm for the \maxlin\ problem on expander graphs.
Our proof technique is inspired by Kolla's algorithm~\cite{kolla11}.
However, in contrast to the algorithm in \cite{kolla11}, 
 we use the Hermitian Laplacian to represent a \maxlin\ instance and show that,
when the underlying graph
has good expansion,
a good approximate solution is encoded in the eigenvector associated with  $\lambda_1(\calL_\ists)$. 
We assume that $G$ is a $d$-regular graph, and hence $\ists = (G,k)$ is a  \maxlin\ instance
with $n$ variables and $nd/2$ equations whose optimum
is $1 - \eps$. One can view $\ists$  as an instance generated by modifying $\eps$ fraction of the constraints~(i.e., edges) from a completely satisfiable instance $\IH = (\GH,k)$. Hence, 
a satisfiable assignment $\psi: V\to [k]$ for $\IH$
will satisfy at least a $(1 - \eps)$-fraction of equations in $\ists$.

Now we discuss the techniques used to  prove Theorem~\ref{thm:res2}.
Let $y_{\psi} \in \mathbb{C}^n$ be the normalised ``indicator vector'' of $\psi$, i.e.,
$\kh{y_{\psi}}_u = \frac{1}{\sqrt{n}} \omega_k^{\psi(u)}$.
Then it holds that 
\begin{align}
    \kh{y_\psi}^* \calL_{\IH} y_{\psi} =
    \frac{1}{d}
    \sum_{u\leadsto v} \weit_{uv} \norm{(y_{\psi})_u - \omega_k^{\cst_{uv}} (y_{\psi})_v}^2
    = 0.\notag
\end{align}
This implies that $y_{\psi}$ is an eigenvector associated with
$\lambda_1\left(\calL_{\IH} \right)=0$.
We denote by $\calU$ the underlying undirected graph of $G$,
and denote by $\calL_{\calU}$ the normalised Laplacian of 
 $\calG$.
Note that since $\calU$ is undirected, $\calL_{\calU}$
only contains real-valued entries. 
We first show that  the eigenvalues of 
$\calL_{\IH}$, the normalised Laplacian of the completely satisfiable instance, and of $\calL_{\calU}$, the normalised Laplacian of the underlining undirected graph $\calU$, coincide. 
Since $\calL_{\calU}$ is the Laplacian matrix of an expander graph, this implies that there is a gap between 
$\lambda_1\left( \calL_{\IH}\right)$ and $\lambda_2\left( \calL_{\IH}\right)$. 

\begin{lemma}\label{lem:lambdas}
It holds for all $1\leq i\leq n$ that 
    $\lambda_i\left(\calL_{\IH}\right) = \lambda_i\left(\calL_{\calU}\right)$.
\end{lemma}

\begin{proof}
    For any unit-norm eigenvector $f_i \in \mathbb{R}^n$ corresponding to the eigenvalue $\lambda_i(\calL_{\calU})$,
    we construct another unit vector $g_i\in\mathbb{C}^n$ such that
     \[
    (g_i)_u = (y_{\psi})_u (f_i)_u = \frac{1}{\sqrt{n}}\ \omega_k^{\psi(u)} \cdot (f_i)_u.
    \]
     Then, it follows that
    \begin{align*}
        \kh{\calL_{\IH} g_i}_u &=
        \frac{1}{d} \kh{ \sum\nolimits_{u\leadsto v} \weit_{uv} \kh{(g_i)_u - \omega_k^{\cst_{uv}} (g_i)_v}
        + \sum\nolimits_{v\leadsto u} \weit_{vu} \kh{(g_i)_u - \conj{\omega_k}\,^{\cst_{uv}}(g_i)_v}} \notag \\
        &=
        \frac{(y_{\psi})_u}{d}
        \kh{ \sum\nolimits_{u\leadsto v} \weit_{uv} \kh{(f_i)_u - (f_i)_v}
        + \sum\nolimits_{v\leadsto u} \weit_{vu} \kh{(f_i)_u - (f_i)_v}} \notag \\
        &= (y_{\psi})_u\cdot  \kh{\calL_{\calG} f_i}_u \notag \\
        &= \lambda_i(\calL_{\calG}) (y_{\psi})_u (f_i)_u \notag \\
        &= \lambda_i(\calL_{\calG}) (g_i)_u,
    \end{align*}
    which implies that $\calL_{\IH} g_i = \lambda_i(\calL_{\calG}) g_i$.
    Since by construction
    $g_i$ is orthogonal to $g_j$ for any $i\neq j$,
    the lemma follows.
\end{proof}

Next we bound the perturbation of the bottom eigenspace of $\calL_{\IH}$ when the latter is turned into $\calL_{\ists}$.  
In particular, 
Lemma~\ref{lem:6.2} below proves that this perturbation does not affect too much to the vectors that have norm spreads out uniformly over all their coordinates.

\begin{lemma}\label{lem:6.2}
    Let $f\in \mathbb{C}^n$ be
    a vector such that $\norm{f_u} = \frac{1}{\sqrt{n}}$
    for all $u\in V$.
    It holds that
    \begin{align}\label{eq:6.1}
        \norm{\kh{\calL_\ists - \calL_{\IH}}f} \leq 2\sqrt{\eps}.
    \end{align}
\end{lemma}

\begin{proof}
Let $R\in\RRR^{n\times n}$ be a matrix defined by
    \begin{align}
        R_{uv} =
        \begin{cases}
            \weit_{uv}/d & \mbox{if}\  \kh{\calL_\ists}_{uv} \neq \kh{\calL_{\IH}}_{uv}, \\
            0 & \mbox{otherwise}.
        \end{cases}\notag
    \end{align}
  Then, it holds that
    \begin{align}\label{eq:sqr}
        \norm{\kh{\calL_\ists - \calL_{\IH}}f} &=
        \sqrt{\sum_{u\in V}
        \norm{\sum_{v\in V} \kh{\calL_\ists - \calL_{\IH}}_{uv} f_v}^2} \notag \\
        &\leq
        \sqrt{\sum_{u\in V}
        \kh{\sum_{v\in V} \norm{\kh{\calL_\ists - \calL_{\IH}}_{uv} f_v}}^2} \notag \\
        &\leq
        \sqrt{\sum_{u\in V}
        \kh{\sum_{v\in V} 2R_{uv}\norm{f_v}}^2} \notag \\
        & =\frac{2}{\sqrt{n}} \sqrt{\sum_{u\in V}
        \kh{\sum_{v\in V} R_{uv}}^2}.
    \end{align}
    Since $\ists$  can be viewed as   modifying  an $\eps$-fraction of the edges from $\IH$, the sum of the entires 
     of $R$ is at most $\eps n d / d = \eps n$,
    and the sum of each row of $R$ is at most $1$.
    Since (\ref{eq:sqr})
    is maximised when there are $\eps n$ rows of $R$ whose sum is $1$,
     we obtain (\ref{eq:6.1}).
\end{proof}

Based on Lemma~\ref{lem:6.2}, we prove that the  change from  $\calL_{\IH}$ to $\calL_\ists$  doesn't have too much influence on the eigenvector associated with  $\lambda_1(\calL_\ists)$. 
For simplicity, let $\lambda_2 = \lambda_2(\calL_{\IH}) = \lambda_2(\calL_{\calG})$.

\begin{lemma}\label{lem:dsk}
    Let $f_1\in \mathbb{C}^n$
    be a unit eigenvector associated with $\lambda_1(\calL_\ists)$.
    Then we have
    \begin{align} \notag
        \norm{\kh{\calL_\ists - \calL_{\IH}}f_1} \leq 20\sqrt{\frac{\eps}{\lambda_2}}.
    \end{align}
\end{lemma}

\begin{proof}
    We first show that $f_1$ is close to a unit vector
    whose coordinates are all of the same norm $\frac{1}{\sqrt{n}}$.
    Let $g\in\mathbb{C}^n$ and $h\in\mathbb{R}^n$ be defined by
    \begin{align} \notag
        g_u = \frac{(f_1)_u}{\sqrt{n}\norm{(f_1)_u}}\quad\text{and}\quad
        h_u = \norm{(f_1)_u}.
    \end{align}
    Then we have
    \begin{align}\label{eq:6.6}
        h^\rot \calL_{\calG} h =
        \frac{1}{d} \sum_{u\leadsto v} \kh{\norm{(f_1)_u} - \norm{(f_1)_v}}^2 \leq
        \frac{1}{d} \sum_{u\leadsto v} \norm{(f_1)_u - \omega_k^{\cst_{uv}} (f_1)_v}^2 =
        f_1^* \calL_\ists f_1 \leq 2\eps,
    \end{align}
    where the last inequality follows from the easy direction of our Cheeger inequality~(Theorem~\ref{thm:cheeger_k}). We introduce parameters $a,b$ such that 
    \begin{align}
        h = a \vec{1} 
        + b \vec{1}_\bot, \notag
    \end{align}
    where $\vec{1}$ is the normalised all-ones (i.e., with all $\frac{1}{\sqrt{n}}$ entries) vector and $\vec{1}_\bot$ is a unit vector orthogonal to $\vec{1}$.      Since $\vec{1}$ is the eigenvector associated with $\lambda_1(\calL_{\calG}) = 0$,
    it holds that
    \begin{align}
        h^{\rot} \calL_{\calG} h = b^2\kh{\kh{\vec{1}_\bot}^\rot \calL_{\calG} \vec{1}_\bot}  \geq b^2\lambda_2 , \notag
    \end{align}
    which coupled with (\ref{eq:6.6}) gives us that  $b \leq \sqrt{\frac{2\eps}{\lambda_2}}$.
    Hence,  we can upper bound the distance between $h$ and $\vec{1}$ by     \begin{align} 
        \norm{h - \vec{1}} =
        \sqrt{(1 - a)^2 + b^2} \leq
        \sqrt{1 - a^2 + b^2}
        = \sqrt{2}b \leq
        2\sqrt{\frac{\eps}{\lambda_2}} \notag
    \end{align}
    where the first inequality holds since
    $h$ is a unit vector and thus $a \in [0,1]$.
    This gives us that  
    \begin{align*}
        \norm{f_1 - g} = \sqrt{\sum_{u\in V} \norm{(f_1)_u - \frac{(f_1)_u}{\sqrt{n}\norm{(f_1)_u}}}^2}
        = \sqrt{\sum_{u\in V} \kh{\norm{(f_1)_u} - \frac{1}{\sqrt{n}}}^2}
        = \norm{h - \vec{1}}\leq 2\sqrt{\frac{\eps}{\lambda_2}}.
    \end{align*}
    We can use this to derive the upper bound in this lemma by
    \begin{align}
        \norm{\kh{\calL_\ists - \calL_{\IH}}f_1} &\leq
        \norm{\kh{\calL_\ists - \calL_{\IH}} g} +
        \norm{\kh{\calL_\ists - \calL_{\IH}} \kh{f_1 - g}} \notag \\
        &\leq 2\sqrt{\eps} + \norm{\calL_\ists \kh{f_1 - g}} + \norm{\calL_{\IH} \kh{f_1 - g}} \notag\\
        &\leq 2\sqrt{\eps} + 4\norm{f_1 - g} \notag\\
        &\leq 20 \sqrt{\frac{\eps}{\lambda_2}}, \notag
    \end{align}
    where the second inequality follows from Lemma~\ref{lem:6.2},
    and the third inequality follows from  the fact  that the eigenvalues
    of  $\calL_\ists$ and $\calL_{\IH}$ are at most $2$.
    \end{proof}

We then prove the following lemma which shows that
the eigenvector $f_1$ corresponding to $\lambda_1(\calL_\ists)$ is close to $y_{\psi}$,
the indicator vector of the optimal assignment $\psi$.

\begin{lemma}\label{lem:span}
    Let $f_1\in \mathbb{C}^n$
    be a unit eigenvector associated with $\lambda_1(\calL_\ists)$. Then, there exist $\alpha,\beta \in \mathbb{C}$ and a unit vector $y_\bot\in\mathbb{C}^n$
    orthogonal to $y_{\psi}$ (i.e. $\kh{y_\bot}^* y_{\psi} = 0$) such that
$ f_1 = \alpha y_{\psi} + \beta y_\bot $
    and
    $
        \norm{\beta} \leq 30 \sqrt{\eps/ \lambda_2^3}$.
\end{lemma}

\begin{proof}
    The proof essentially corresponds to the Davis-Kahan theorem~\cite{DavisKahan} for 1-dimensional eigenspaces. Let $y_{\psi} = v_1,\ldots,v_n\in\mathbb{C}^n$ be the orthonormal eigenvectors
    associated with eigenvalues $0 = \lambda_1\left(\calL_{\IH}\right) \le \cdots \le \lambda_n\left(\calL_{\IH}\right)$ of $\calL_{\IH}$, 
    which means $\calL_{\IH}$ can be diagonalised by
    $\calL_{\IH} = \sum_{i=2}^n \lambda_i\left(\calL_{\IH}\right) v_1 v_1^*$.
   Then it holds that 
    \begin{align}\label{eq:6.14}
        \norm{\calL_{\IH} f_1}^2 =
        f_1^{\rot} \calL_{\IH}^2 f_1
        = \sum_{i=2}^n \lambda_i^2\left(\calL_{\IH}\right) \norm{f_1^* v_i}^2
        \geq \lambda_2^2\left(\calL_{\IH}\right) \kh{1 - \norm{f_1^* y_{\psi}}^2}
        = \lambda_2^2 \norm{\beta}^2.
    \end{align}
    By Lemma~\ref{lem:dsk},
    the square root of this quantity can be upper bounded by
    \begin{align}\label{eq:6.15}
        \norm{\calL_{\IH} f_1}
        \leq
        \norm{\calL_{\ists} f_1} +
        \norm{\kh{\calL_{\ists} - \calL_{\IH}} f_1}
        \leq \lambda_1\kh{\calL_{\ists}} + 20 \sqrt{\frac{\eps}{\lambda_2}}
        \leq 30\sqrt{\frac{\eps}{\lambda_2}},
    \end{align}
    where the last inequality follows by noting
    $\lambda_1(\calL_{\ists}) \leq 2\eps$ by the easy direction
    of our Cheeger inequality and $\lambda_2\leq 2$.
    Combining~(\ref{eq:6.14}) and~(\ref{eq:6.15}) proves the statement.
\end{proof}


Based on  Lemma~\ref{lem:span}, 
$f_1$ is close to the indicator vector of an optimal assignment rotated by some angle.
In particular, we have that 
\begin{align} \label{eq:abc}
    \norm{f_1 - \frac{\alpha}{\norm{\alpha}} y_{\psi}} =
    \sqrt{(1 - \norm{\alpha})^2 + \norm{\beta}^2} \leq
    \sqrt{1 - \norm{\alpha}^2 + \norm{\beta}^2}
    = \sqrt{2}\norm{\beta} \leq
    30\sqrt{\frac{2\eps}{\lambda_2^3}},
\end{align}
where $\frac{\alpha}{\norm{\alpha}} y_{\psi}$ is the vector that encodes the information of
an assignment that satisfies
all the  equations in $\IH$
and at least $1 - \eps$ fraction
of equations in $\ists$. 
Therefore, our goal is to recover $\frac{\alpha}{\norm{\alpha}} y_{\psi}$
from $f_1$. 

\begin{proof}[Proof of Theorem~\ref{thm:res2}]
    Let $\psi$ be the optimal assignment of $\ists$
    satisfying $1 - \eps$ fraction of equations,
    which is also a completely satisfying assignment of
    $\IH$.
    Let $f_1$ be a unit eigenvector associated with $\lambda_1(\calL_{\ists})$.
    By Lemma~\ref{lem:span},
    there exists $\alpha,\beta\in \mathbb{C}$ such that
    $f_1 = \alpha y_{\psi} + \beta y_{\bot}$
    where $\norm{\beta} \leq 30 \sqrt{\eps /\lambda_2^3}$.
    Our goal is to find a
    vector $z_\asn \in \mathbb{C}^n$,
    which equals the indicator vector of $\asn$ ratoted by some angle
    and satisfies
    \begin{align}\label{eq:sat}
        \norm{f_1 - z_\asn } \leq
        \norm{f_1  - \frac{\alpha}{\| \alpha\|} y_{\psi}}  \leq
        30\sqrt{\frac{2\eps}{\lambda_2^3}},
    \end{align}
    where the last inequality follows by (\ref{eq:abc}).
    The assignment $\asn$ corresponding to
    such a $z_\asn$ will give us
    that
    the fraction of unsatisfied equations by
    $\asn$ is     
    \begin{align}
        \pnt^{\asn}(\ists) &\leq 10 k^2 z_{\asn}^* \calL_{\ists} z_{\asn} \notag \\
        &= 10k^2 (z_{\asn} - f_1 + f_1)^* \calL_{\ists} (z_{\asn} - f_1 + f_1) \notag\\
        &\leq k^2 \kh{ (z_\asn - f_1)^* \calL_\ists (z_\asn - f_1)
        + f_1^* \calL_\ists f_1
    + 2 \norm{ (z_\asn - f_1)^* \calL_\ists f_1} }
        \notag\\
        &\leq 10k^2 \kh{ 2 \norm{z_\asn - f_1}^2 + \lambda_1(\calL_\ists) +
        2 \norm{z_\asn - f_1} \sqrt{\lambda_1{(\calL_\ists)}} } \notag \\
        & \leq 10k^2 \left( 2\cdot 900\cdot \frac{2\varepsilon}{\lambda_2^3} + 2\varepsilon + 2\cdot 30\cdot \sqrt{\frac{2\varepsilon}{\lambda_2^3} }\cdot
       \sqrt{2\varepsilon} \right) \notag\\
        &\leq 100000 k^2 \cdot \frac{\eps}{\lambda_2^3}, \notag
    \end{align}
    where the factor $10 k^2$ above follows from the fact that
    $\norm{1 -\omega_k^j}^2$ is at least $1/(10k^2)$
    for $j = 1,\ldots,k-1$.

    To find such vector $z_{\phi}$ satisfying (\ref{eq:sat}), we define  $\asn_\eta: V \rightarrow [k]$ by
    \begin{align}\notag
        \asn_\eta(u) = \mathrm{arg\,min}_{j\in[k]} \norm{(f_1)_u - \mathrm{e}^{\eta i}\omega_k^j}.
    \end{align}
    Notice that, since $\frac{\alpha}{\norm{\alpha}}$ is equal to $\mathrm{e}^{\eta i}$ for some $\eta\in [0,2\pi)$, by defining
    $
        (z_{\asn_\eta})_u = \mathrm{e}^{\eta i} \omega_k^{\asn_\eta(u)}
    $ the solution to the following optimisation problem 
    \begin{align}\notag
        \min_{\eta\in [0,2\pi)}  \norm{z_{\asn_\eta} - f_1}
    \end{align}
    gives us a vector that satisfies (\ref{eq:sat}).
               To solve this optimisation problem,
    we notice that it suffices to consider $\eta$ in the range $[0, 2\pi / k)$.
    Therefore, we simply enumerate all $\eta$'s over the following discrete set:
    \begin{align}\notag
    \setof{\frac{t\sqrt{\eps}}{\sqrt{n}}\ \left|\ t = 0,1,\ldots,\ceil{\frac{2\pi \sqrt{n}}{k \sqrt{\eps}}}\right.}.
    \end{align}
    By enumerating this set, we can find an assignment $\asn$ and an $\eta$ such that
    \begin{align}\notag
        \norm{f_1 - z_{\asn_\eta} } \leq \norm{ f_1 - \frac{\alpha }{\| \alpha \|}y_{\psi} } + O(\sqrt{\eps}),
    \end{align}
    which is enough to get our desired approximation.
    Since the size of this set is $O\left(\frac{\sqrt{n}}{k\sqrt{\eps}}\right)$,
    the total running time is $O\left(\frac{n^{1.5}}{k\sqrt{\eps}}\right)$
    plus the running time needed to compute the eigenvector $f_1$.
\end{proof}

\section{Concluding remarks}

Our work leaves several open questions for further research: while the factor of $k$ in our Cheeger inequality (Theorem~\ref{thm:cheeger_k}) is needed, it would be interesting to see if it's possible to construct a different Laplacian for which a similar Cheeger inequality holds with a smaller dependency on $k$. 
For example, instead of embedding vertices in $\CCC$ and mapping assignments to roots of unity, one could consider embedding vertices in higher dimensions using the  bottom $k$ eigenvectors of the Laplacian of the label extended graph, and see  if a relation between the imperfectness ratio of Definition~\ref{def:perfect} and the $k$-th smallest eigenvalue of this Laplacian still holds.


Finally, we observe that several cut problems in directed graphs can be formulated as special cases of \maxlin\ (see, e.g., \cite{AnderssonEH01,GW04}). Because of this, we believe the Hermitian Laplacians studied in our paper will have further applications in the development of fast algorithms for combinatorial problems on directed graphs, and might have further connections to Unique Games.

\paragraph{Acknowledgement.} 
The project is partially supported by the ERC Starting Grant (DYNAMIC MARCH).
We are very grateful to Mihai Cucuringu for valuable discussions on Hermitian Laplacian matrices and their applications. We would also like to thank Chris Heunen for some fruitful conversations on topics closely related to this paper.






\appendix

\bibliographystyle{alpha}
\bibliography{references}

\newcommand{\etalchar}[1]{$^{#1}$}
\begin{thebibliography}{KKMO07}

\bibitem[ABS15]{abs}
Sanjeev Arora, Boaz Barak, and David Steurer.
\newblock Subexponential algorithms for unique games and related problems.
\newblock {\em Journal of the {ACM}}, 62(5), 2015.

\bibitem[Ach03]{Ach03}
Dimitris Achlioptas.
\newblock Database-friendly random projections: Johnson-{L}indenstrauss with
  binary coins.
\newblock {\em Journal of computer and System Sciences}, 66(4):671--687, 2003.

\bibitem[AEH01]{AnderssonEH01}
Gunnar Andersson, Lars Engebretsen, and Johan H{\aa}stad.
\newblock A new way of using semidefinite programming with applications to
  linear equations mod $p$.
\newblock {\em Journal of Algorithms}, 39(2):162--204, 2001.

\bibitem[AKK{\etalchar{+}}08]{stoc/AroraKKSTV08}
Sanjeev Arora, Subhash Khot, Alexandra Kolla, David Steurer, Madhur Tulsiani,
  and Nisheeth~K. Vishnoi.
\newblock Unique games on expanding constraint graphs are easy.
\newblock In {\em \STOC{40th}{08}}, pages 21--28, 2008.

\bibitem[Alo86]{Alon86}
Noga Alon.
\newblock Eigenvalues and expanders.
\newblock {\em Combinatorica}, 6(2):83--96, 1986.

\bibitem[BSS13]{BandeiraSS13}
Afonso~S. Bandeira, Amit Singer, and Daniel~A. Spielman.
\newblock A {C}heeger inequality for the graph connection {L}aplacian.
\newblock {\em {SIAM} Journal on Matrix Analysis and Applications},
  34(4):1611--1630, 2013.

\bibitem[Chu97]{chung1}
Fan R.~K. Chung.
\newblock Spectral graph theory.
\newblock {\em Regional Conference Series in Mathematics, American Mathematical
  Society}, 92:1--212, 1997.

\bibitem[CKM{\etalchar{+}}14]{CKM+14}
Michael~B. Cohen, Rasmus Kyng, Gary~L. Miller, Jakub~W. Pachocki, Richard Peng,
  Anup~B. Rao, and Shen~Chen Xu.
\newblock Solving {SDD} linear systems in nearly $mlog^{1/2}n$ time.
\newblock In {\em \STOC{46th}{14}}, pages 343--352, 2014.

\bibitem[CMM06]{charikar06}
Moses Charikar, Konstantin Makarychev, and Yury Makarychev.
\newblock Near-optimal algorithms for unique games.
\newblock In {\em \STOC{38th}{06}}, pages 205--214, 2006.

\bibitem[DK70]{DavisKahan}
Chandler Davis and William~M. Kahan.
\newblock The rotation of eigenvectors by a perturbation. iii.
\newblock {\em SIAM Journal on Numerical Analysis}, 7(1):1--46, 1970.

\bibitem[FL92]{feigelovasz92}
Uriel Feige and L{\'a}szl{\'o} Lov{\'a}sz.
\newblock Two-prover one-round proof systems: Their power and their problems.
\newblock In {\em \STOC{24th}{92}}, pages 733--744, 1992.

\bibitem[FR04]{feige04}
Uriel Feige and Daniel Reichman.
\newblock On systems of linear equations with two variables per equation.
\newblock In {\em \APPROX{7th}{04}}, pages 117--127, 2004.

\bibitem[GT06]{gupta06}
Anupam Gupta and Kunal Talwar.
\newblock Approximating unique games.
\newblock In {\em \SODA{17th}{06}}, pages 99--106, 2006.

\bibitem[GW95]{GW95}
Michel~X. Goemans and David~P. Williamson.
\newblock Improved approximation algorithms for maximum cut and satisfiability
  problems using semidefinite programming.
\newblock {\em Journal of the {ACM}}, 42(6):1115--1145, 1995.

\bibitem[GW04]{GW04}
Michel~X. Goemans and David~P. Williamson.
\newblock Approximation algorithms for {Max-3-Cut} and other problems via
  complex semidefinite programming.
\newblock {\em Journal of Computer and System Sciences}, 68(2):442--470, 2004.

\bibitem[H{\aa}s01]{Hastad01}
Johan H{\aa}stad.
\newblock Some optimal inapproximability results.
\newblock {\em Journal of the {ACM}}, 48(4):798--859, 2001.

\bibitem[JL84]{JL84}
William~B Johnson and Joram Lindenstrauss.
\newblock Extensions of {L}ipschitz mappings into a {H}ilbert space.
\newblock {\em Contemporary mathematics}, 26(189-206):1, 1984.

\bibitem[Kar72]{Karp72}
Richard~M. Karp.
\newblock Reducibility among combinatorial problems.
\newblock In {\em a symposium on the Complexity of Computer Computations},
  pages 85--103, 1972.

\bibitem[Kho02]{Kho02a}
Subhash Khot.
\newblock On the power of unique 2-prover 1-round games.
\newblock In {\em \STOC{34th}{02}}, pages 767--775, 2002.

\bibitem[KKMO07]{KhotKMO07}
Subhash Khot, Guy Kindler, Elchanan Mossel, and Ryan O'Donnell.
\newblock Optimal inapproximability results for {MAX-CUT} and other 2-variable
  {CSP}s?
\newblock {\em {SIAM} Journal on Computing}, 37(1):319--357, 2007.

\bibitem[Kle99]{Kleinberg99}
Jon~M. Kleinberg.
\newblock Authoritative sources in a hyperlinked environment.
\newblock {\em Journal of the ACM}, 46(5):604--632, 1999.

\bibitem[KLL{\etalchar{+}}13]{KwokLLGT13}
Tsz~Chiu Kwok, Lap~Chi Lau, Yin~Tat Lee, Shayan~Oveis Gharan, and Luca
  Trevisan.
\newblock Improved {C}heeger's inequality: analysis of spectral partitioning
  algorithms through higher order spectral gap.
\newblock In {\em \STOC{45th}{13}}, pages 11--20, 2013.

\bibitem[KLP{\etalchar{+}}16]{kyng16}
Rasmus Kyng, Yin~Tat Lee, Richard Peng, Sushant Sachdeva, and Daniel~A.
  Spielman.
\newblock Sparsified {C}holesky and multigrid solvers for connection
  {L}aplacians.
\newblock In {\em \STOC{48th}{16}}, pages 842--850, 2016.

\bibitem[Kol11]{kolla11}
Alexandra Kolla.
\newblock Spectral algorithms for unique games.
\newblock {\em Computational Complexity}, 20(2):177--206, 2011.

\bibitem[KS16]{KyngS16}
Rasmus Kyng and Sushant Sachdeva.
\newblock Approximate {G}aussian elimination for {L}aplacians - fast, sparse,
  and simple.
\newblock In {\em \FOCS{57th}{16}}, pages 573--582, 2016.

\bibitem[LGT14]{journals/jacm/LeeGT14}
James~R. Lee, Shayan~Oveis Gharan, and Luca Trevisan.
\newblock Multiway spectral partitioning and higher-order {C}heeger
  inequalities.
\newblock {\em Journal of the {ACM}}, 61(6):37:1--37:30, 2014.

\bibitem[PSZ17]{PSZ17}
Richard Peng, He~Sun, and Luca Zanetti.
\newblock Partitioning well-clustered graphs: Spectral clustering works!
\newblock {\em {SIAM} Journal on Computing}, 46(2):710--743, 2017.

\bibitem[Sin11]{singer11}
Amit Singer.
\newblock Angular synchronization by eigenvectors and semidefinite programming.
\newblock {\em Applied and computational harmonic analysis}, 30(1):20, 2011.

\bibitem[SM00]{ShiM00}
Jianbo Shi and Jitendra Malik.
\newblock Normalized cuts and image segmentation.
\newblock {\em {IEEE} Transactions on Pattern Analysis and Machine
  Intelligence}, 22(8):888--905, 2000.

\bibitem[SS11]{SpielmanS11}
Daniel~A. Spielman and Nikhil Srivastava.
\newblock Graph sparsification by effective resistances.
\newblock {\em SIAM Journal on Computing}, 40(6):1913--1926, 2011.

\bibitem[ST14]{ST14}
Daniel~A. Spielman and Shang{-}Hua Teng.
\newblock Nearly linear time algorithms for preconditioning and solving
  symmetric, diagonally dominant linear systems.
\newblock {\em {SIAM} Journal on Matrix Analysis and Applications},
  35(3):835--885, 2014.

\bibitem[Tre05]{trevisan05}
Luca Trevisan.
\newblock Approximation algorithms for unique games.
\newblock In {\em \FOCS{46th}{05}}, pages 197--205, 2005.

\bibitem[Tre12]{Trevisan09}
Luca Trevisan.
\newblock Max cut and the smallest eigenvalue.
\newblock {\em {SIAM} Journal on Computing}, 41(6):1769--1786, 2012.

\bibitem[Tro12]{Tro12}
Joel~A. Tropp.
\newblock User-friendly tail bounds for sums of random matrices.
\newblock {\em Foundations of Computational Mathematics}, 12(4):389--434, 2012.

\bibitem[Vis13]{vishnoilx}
Nisheeth~K. Vishnoi.
\newblock ${L}x= b$.
\newblock {\em Foundations and Trends in Theoretical Computer Science},
  8(1--2):1--141, 2013.

\end{thebibliography}

\end{document}